\documentclass{llncs}

\usepackage{xcolor}

\usepackage{amsmath}
\usepackage{bcprules}
\usepackage{multicol}
\usepackage{listings,jvlisting}
\usepackage{tikz}
\usepackage{thmtools}
\usepackage{thm-restate}
\usepackage{stmaryrd}
\usepackage{proof}
\usepackage{etoolbox}
\usepackage{caption}
\usepackage{url}
\usepackage{float}
\usepackage{graphicx}
\usepackage{pgfplots}
\usepackage{cite}
\usepackage{mathtools}
\usepackage{subfig}

\pgfplotsset{
/pgfplots/aan ybar legend/.style={
/pgfplots/legend image code/.code={
\draw [##1,/tikz/.cd,bar width=3.5pt,yshift=-0.3em]
plot coordinates {(0cm,0.8em) };}}}

\lstdefinestyle{mystyleC}{
    language=C,                
    basicstyle={\footnotesize\ttfamily\bfseries},
    morekeywords={hls::stream, StreamData},
    keywordstyle=\color{orange},
    commentstyle=\color{blue},
    stringstyle=\color{red},
    keywords=[2]{filter, read, write, mkarray, mkstream, divide, divide_rev, copy, initialize},
    keywordstyle=[2]{\color{violet}},
    numbers=left,              
    numberstyle=\tiny,        
    stepnumber=1,              
    numbersep=1pt,             
    frame=none,              
    tabsize=2,                 
    breaklines=true,           
    showstringspaces=false,    
}

\newcommand{\CC}{C\nolinebreak\hspace{-.05em}\raisebox{.4ex}{\tiny\bf +}\nolinebreak\hspace{-.10em}\raisebox{.4ex}{\tiny\bf +}}
\newcommand\bte[1]{\omit\rlap{\textcolor{blue}{#1}}}
\newcommand\Out{\mathit{out}}
\newcommand\In{\mathit{in}}
\newcommand\hd{\mathit{head}}
\newcommand\tl{\mathit{tail}}
\newcommand\bj[5]{#1\p \{#2\}#3\To #5\{#4\}}
\newcommand\bjs[4]{\bj{#1}{#2}{\bullet}{#3}{#4}}
\newcommand\sj[5]{#1\p \{#2\}#3\To #5\{#4\}}
\newcommand\iset[1]{\iota_{#1}}
\newcommand\Iset[2]{\mathit{Idx}(#1,#2)}
\newcommand\Iseq{\Theta}
\newcommand\setof[1]{\mathit{setof}(#1)}
\newcommand\length[1]{|#1|}

\newcommand\ATE{\Gamma}
\newcommand\BTE{\mathcal{B}}

\newcommand\forexp[5]{\mathbf{for}(#1 := #2; #1 \ne #3; #1 := #1+#4) #5}
\newcommand\whileexp[2]{\mathbf{while}\ #1\ \mathbf{do}\ #2}

\newcommand\RARR{\mathbf{rarray}}
\newcommand\WARR{\mathbf{warray}}

\newcommand\Strm{S}
\newcommand\iexp{u} 
\newcommand\texp{t} 
\newcommand\RStr[1]{{#1}.read()}
\newcommand\WStr[2]{{#1}.write(#2)}

\newcommand\BeginKer{[}
\newcommand\EndKer{]_K}
\newcommand\kernelexp[1]{[ #1 ]_K}
\newcommand\eval{\Downarrow}

\newcommand\set[1]{\{#1\}}
\newcommand\Imp{\Rightarrow}
\newcommand\COL{\mathbin{:}}
\newcommand\INT{\mathbf{int}}

\newcommand\BUF{\mathbf{buf}}

\newcommand\p{\vdash}

\newcommand\ty{\tau}
\newcommand\To{\Longrightarrow}

\newcommand\dom{\mathit{dom}}

\newcommand\ifexp[2]{\mathbf{if}\ #1\ \mathbf{then}\ #2\ \mathbf{else}\ }
\newcommand\OP{\mathtt{op}}

\newcommand\red{\longrightarrow}

\newcommand\TRUE{\mathbf{true}}

\newcommand\sconfig[3]{\langle #1, #2, #3 \rangle}
\newcommand\econfig[2]{\langle #1, #2 \rangle}
\newcommand\sstate[2]{\langle #1, #2 \rangle}
\newcommand\fliparray[1]{#1^\sharp}
\newcommand\flip{\fliparray}

\newcommand*\defeq{\stackrel{\text{def}}{=}}

\newcommand\Addr[2]{(#1,#2)}

\newcommand\seq[1]{\widetilde{#1}}
\newcommand\relRH[6]{\sstate{#1}{#2}\sim_{#5\mid#6}\sstate{#3}{#4}}
\newcommand\updateAr[3]{#1\set{#2\mapsto #3}}

\newcommand*{\low}{\mathit{low}}
\newcommand*{\high}{\mathit{high}}
\newcommand*{\awp}[2]{\mathrm{awp}(#1, #2)}
\newcommand*{\Inv}{\mathit{Inv}}

\newif\ifdraft\draftfalse
\newif\iffull\fulltrue
\ifdraft
\newcommand\COM[1]{\textcolor{red}{[#1]}}
\newcommand\nk[1]{\textcolor{blue}{[#1 -nk]}}
\newcommand\sk[1]{\textcolor{teal}{{\footnotesize [#1 -ks]}}}

\newcommand\Revision[1]{\textcolor{red}{#1}}
\else

\newcommand\COM[1]{}
\newcommand\nk[1]{}
\newcommand\sk[1]{}
\newcommand\Revision[1]{#1}
\fi

\iffull
\newcommand\CameraReady[2]{#2}
\else
\newcommand\CameraReady[2]{#1}
\fi

\begin{document}

\title{Relational Hoare Logic for High-Level Synthesis of Hardware Accelerators}

\author{Izumi Tanaka \and Ken Sakayori \and Shinya Takamaeda-Yamazaki \and Naoki Kobayashi}
\institute{The University of Tokyo, Japan\\
\email{\{i.tanaka,sakayori,shinya,koba\}@is.s.u-tokyo.ac.jp}}

\maketitle

\begin{abstract}
    High-level synthesis (HLS) is a powerful tool for developing efficient hardware accelerators
that rely on specialized memory systems to achieve sufficient on-chip data reuse and off-chip bandwidth utilization.
However, even with HLS, designing such systems still requires careful manual tuning,
as automatic optimizations provided by existing tools are highly sensitive to programming style
and often lack transparency.
To address these issues, 
we present a formal translation framework based on
relational Hoare logic, which enables robust and transparent transformations.
Our method recognizes complex memory access patterns in naïve
HLS programs and automatically transforms them by inserting on-chip buffers 
to enforce linear access to off-chip memory, and by replacing non-sequential processing with stream processing,
while preserving program semantics.
Experiments using our prototype translator, combined with an off-the-shelf
HLS compiler and a real FPGA board, have demonstrated significant performance improvements.

\end{abstract}

\section{Introduction}
\label{sec:intro}

With the end of Moore's Law approaching, 
advances in computing performance increasingly rely on architectural innovation rather than
continued transistor scaling~\cite{hennessy2021golden}.
Domain-specific hardware accelerators have been recognized as an effective solution
for high-performance and energy-efficient computing, particularly in fields such as machine learning~\cite{jouppi2017datacenter}, 
image processing~\cite{hajirassouliha2018suitability}, and other specialized applications~\cite{turakhia2018darwin,sen2013designing,lin2018architectural}.
Among various platforms---such as general-purpose processors, graphics processing units (GPUs), 
and application-specific integrated circuits (ASICs)---field-programmable gate arrays (FPGAs) offer a compelling balance of performance, 
flexibility, and energy efficiency, making them well-suited for developing customized accelerators.

High-level synthesis (HLS)~\cite{canis2011legup,cong2011high,xilinx2020_hls,pilato2013bambu,intel_hls_compiler}, 
which translates programs written in high-level languages such as \CC{} into hardware description languages (HDLs),
is a powerful method for developing application- or domain-specific hardware accelerators on FPGAs. 
This approach enables software developers to generate hardware without requiring in-depth knowledge of low-level design.

\begin{figure}
   \begin{lstlisting}[style=mystyleC]
 // Infers burst transfers for array accesses
 void divide(int* input, int* output) {
   for (int i = 0; i < N; i++) {
     output[i] = input[i] / 2;
   }
 }
   \end{lstlisting}
   \begin{lstlisting}[style=mystyleC]
 // Infers a single transfer for each word of arrays
 void divide_rev(int* input, int* output) {
   for (int i = N-1; i >= 0; i--) {
     output[i] = input[i] / 2;
   }
 }
  \end{lstlisting}
 \caption{Pitfall of Vitis HLS: The first (top) uses burst transfers, whereas the second (bottom) fails to do so.}
  \label{fig:intro-pitfall}
\end{figure}

Despite the abstraction provided by HLS, writing high-performance code remains challenging. 
Developers often need to manually insert annotations such as pragmas and carefully restructure their code to guide the compiler's optimization process.
In particular, achieving high-performance hardware acceleration depends on 
specialized memory systems that enable sufficient on-chip data reuse and efficient off-chip bandwidth utilization, 
both of which typically require careful manual optimization in HLS code.
While existing HLS tools can automatically apply a variety of optimizations, 
these optimizations often lack robustness, as their effectiveness heavily depends on the specific programming style.
For example, Vitis HLS~\cite{vitis_amd}, one of the most widely used HLS tools, 
can apply the burst data transfer optimization when memory is accessed consecutively in increasing order, 
as in the first program at the top of Figure~\ref{fig:intro-pitfall}. 
When this optimization is applied, the synthesized hardware can transfer 
multiple consecutive data between on-chip and off-chip memory in a single transaction, significantly improving performance;
indeed, our experimental results (presented in Section~\ref{sec:experiment}) show 
that the synthesized hardware runs more than five times faster than that without burst transfers.
However, the optimization is not applied to the second program at the bottom of Figure~\ref{fig:intro-pitfall}, 
where the memory is accessed in reverse order.
Although the two programs perform the same computation, the reverse-access version prevents Vitis HLS from recognizing the consecutive access pattern.
This example highlights how minor syntactic differences can lead to significant performance gaps.
Furthermore, because the internal optimization strategies of these tools are insufficiently documented
\Revision{and lack a formal specification of the underlying program transformations}, 
it is often unclear how to modify inefficient programs to achieve better results.

These observations indicate three major limitations of existing HLS tools:
(i) they require manual annotations and code restructuring to achieve high performance;
(ii) their automatic optimizations are sensitive to superficial variations; and
(iii) their optimization strategies lack transparency, making performance tuning difficult.

To address these issues, we propose a formally defined automatic translation framework for HLS programs,
grounded in programming language techniques, specifically \emph{relational Hoare logic}.
Hoare logic~\cite{10.1145/363235.363259} has been used as \emph{the} standard program logic for
proving the correctness of programs. It is robust in that it enables precise reasoning about
program behavior, regardless of the programming style, as guaranteed by the relative completeness theorem~\cite{doi:10.1137/0207005}. 
Relational Hoare logic~\cite{DBLP:conf/isola/Naumann20} is an extension of Hoare logic for reasoning about
the relationship between two programs. It has recently been used for reasoning about probabilistic programs and security properties~\cite{barthe2009formal, DBLP:journals/pacmpl/AvanziniBDG25,DBLP:conf/lics/NagasamudramN21}.
Here, we adopt and tailor relational Hoare logic to formalize program transformations for HLS,
leveraging well-established verification techniques to enable sound and automated translation.
This approach allows our method to recognize complex memory access patterns and systematically transform them,
ensuring robustness against variations in programming style.
For instance, the programs in Figure~\ref{fig:intro-pitfall} can be uniformly transformed into a form that enables burst data transfers,
regardless of their access order.
Another advantage of our formally defined transformation is its transparent and predictable behavior, in contrast to existing HLS tools.

\begin{figure}
  \begin{lstlisting}[style=mystyleC]
  // Infers a single transfer for each word of arrays
  void filter(int* input, int* output) {
    for (int i = 0; i < N-1; i++) {
      output[i] = (input[i] + input[i+1]) / 2;
    }
  }
  
  int main() {
    int input[N], output[N-1];
    for (int i = 0; i < N; i++) {
      input[i] := i;
    }
    filter(input, output);
    ...
	}

  \end{lstlisting}

  \begin{lstlisting}[style=mystyleC,numbersep=2pt]
  // Enables burst transfers via streaming interfaces
  typedef hls::axis<int, 0, 0, 0> StreamData;
  void filter(hls::stream<StreamData> &input,
              hls::stream<StreamData> &output) {
    int b0 = input.read();
    for (int i = 0; i < N-1; i++) {
      int b1 = input.read();
      output.write((b0 + b1) / 2);
      b0 = b1;
    } 
  }

  int main() {
		hls::stream<StreamData> input, output;
		for (int i = 0; i < N; i++) {
			input.write(i);
		}
		filter(input, output);
		... 
  }
  \end{lstlisting}
    \captionof{figure}{Naïve filtering program without burst transfers (top) vs. translated program with burst transfers (bottom).}
    \label{fig:intro1}
\end{figure}

We now present a brief overview of our translation method using a running example given in Figure~\ref{fig:intro1}.
The figure shows a naïve filtering program (top) and its translated counterpart (bottom).
Our translation relation takes the form \(\vdash \{ \BTE \} s \Longrightarrow t \{\BTE'\} \), which roughly means that, with the additional information specified as the pre-condition \( \BTE \), the source program \( s \) can be translated into the target program \( t \) (and the post-condition \(\BTE'\) holds after the execution of these programs).
A notable aspect of our formalization is that the pre- and post-conditions may refer to the access pattern of an array as a sequence of array indices.
For example, in line 13 of the naïve program, the pre-condition takes the form \( \{ \cdots \land \iset{input}=0 \cdot 1\cdots (N-1) \land \cdots \} \).
Here \( \iset{input}\) is a special variable to denote the access pattern of \texttt{input} and the equality indicates that the caller has \emph{written} to \texttt{input} in the order \(0\cdot 1\cdots (N-1)\).
The basic idea of the translation is to replace the array \texttt{input} to a stream if we can deduce that the array \texttt{input} is \emph{read} in this order in the body of \texttt{filter}.
Unfortunately, this is not the case because the read access pattern of \texttt{input} is \( 0 \cdot 1 \cdot 1\cdot 2\cdot 2 \cdots (N - 1) \).
To handle this kind of situation, we have a rule for on-chip buffer insertion 
so that an implicit non-linear array access to off-chip memory can be replaced by a read from the on-chip buffer.
Lines 5 and 7 of the second program are results of the buffer insertion rule.
After the buffer insertions, the duplications in the read access pattern of \texttt{input} are removed, and the access pattern becomes \( 0\cdot 1 \cdots (N - 1) \).
This allows us to translate \texttt{input} into a stream, and enables existing HLS tools to infer burst data transfers.


Our framework supports automatic translation---without requiring any manual annotations---from naïve programs
such as the first program in Figure~\ref{fig:intro1}.
This is achieved by using standard automated verification techniques based on Hoare logic and verification condition generation.
To enable automation,
we restrict the shape of the access pattern to \(i\cdot(i+k)\cdot(i+2k)\cdots (i+mk)\).
This restriction allows each access pattern to be concisely represented as a triple \((i,j,k)\) (where \(j=i+mk\)),
which in turn simplifies the verification conditions that are to be discharged by SMT solvers.
Despite this restriction, our method can still handle a wide range of access patterns---including reversed, overlapping, and non-consecutive accesses---as confirmed by our experiments.


We have implemented a prototype tool that automates the translation based on the formalization.
Figure~\ref{fig:overview} illustrates the tool architecture.
Our tool takes as input a source program written in a C-like language, which includes both 
a \emph{kernel function}\Revision{---a function designated to be synthesized into a hardware accelerator by an HLS tool---}and 
other components, such as a main function.
(The other components, commonly called \emph{host code}, manage the execution of the accelerator on the FPGA from the host CPU.)
The tool then performs translation (when applicable), and generates an optimized version of the kernel function and host code.
For our running example, the input program is at the top of Figure~\ref{fig:intro1}, 
and the output kernel and host codes are shown at the bottom of the same figure.
The kernel function is then compiled into hardware description language (HDL) using
Vitis HLS~\cite{vitis_amd} and synthesized into hardware with Vivado~\cite{vivado_amd}. 

As discussed in Section~\ref{sec:concl}, our framework could potentially be applied to various program analyses and transformations, 
such as identifying when tiling is required and coordinating multiple functional units through formally verified scheduling. 
These applications are made possible by the expressive reasoning power of relational Hoare logic, 
which allows flexible yet provably correct analyses and transformations. 

\usetikzlibrary{shapes.geometric, arrows, fit}

\tikzstyle{node} = [rectangle, rounded corners, minimum width=1.3cm, minimum height=0.8cm, text centered, draw=black, font=\small]
\tikzstyle{ourtool} = [rectangle, rounded corners, draw=blue, thick, inner sep=0.25cm]
\tikzstyle{arrow} = [thick,->,>=stealth]
\tikzstyle{dasharrow} = [thick,->,>=stealth, dashed]

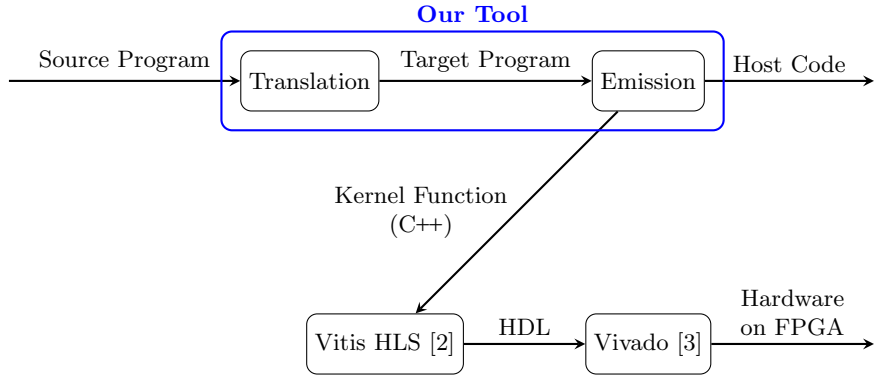
\begin{figure*}[t]
    \centering
    \begin{tikzpicture}[node distance=2cm]

        \node (buffer) [node, xshift=2.5cm, align=center] {Translation};
        \node (emission) [node, right of=buffer, xshift=2.5cm, align=center] {Emission};
        \node (vivado) [node, below of=emission, yshift=-1.5cm, align=center] {Vivado~\cite{vivado_amd}};
        \node (vitis) [node, left of=vivado, xshift=-1.5cm, align=center] {Vitis HLS~\cite{vitis_amd}};
        
        \draw [arrow] (-1.5,0) -- node[anchor=south, align=center, font=\small] {Source Program} (buffer);
        \draw [arrow] (buffer) -- node[anchor=south, align=center, font=\small] {Target Program} (emission);
        \draw [arrow] (emission) --  node[anchor=south, align=center, font=\small] {Host Code} ++(3,0);
        \draw [arrow] (emission) --  node[anchor=east, align=center, font=\small] {Kernel Function\\(\CC{})} (vitis);
        \draw [arrow] (vitis) --  node[anchor=south, align=left, font=\small] {HDL} (vivado);
        \draw [arrow] (vivado) --  node[anchor=south, align=center, font=\small] {Hardware\\on FPGA} ++(3,0);
    
        \node [ourtool, fit=(buffer) (emission), label={[text=blue]above:\textbf{Our Tool}}] {};
    
    \end{tikzpicture}
    \caption{High-level synthesis toolchain with our tool.}
    \label{fig:overview}

\end{figure*}

The contributions of this paper are summarized as follows.
\begin{itemize}
\item Formalization of the program translation based on relational Hoare logic,
  and proofs of their correctness. To our knowledge, we are the first to
  employ relational Hoare logic to formalize HLS-level program transformations.
\item Automation of the translation by a suitable restriction on the shape of access patterns.
\item Implementation of the proposed method and experiments. Our experiments confirm the effectiveness of
  our approach: it achieves significant performance improvement in cases where Vitis HLS fails to optimize.
\end{itemize}

The rest of this paper is organized as follows. 
Section~\ref{sec:lang} introduces the source and target languages.
Section~\ref{sec:buffer} presents the formalization of our translation method.
Section~\ref{sec:auto} describes how to automate the translation.
Section~\ref{sec:experiment} reports experimental results.
Section~\ref{sec:related} discusses related work, and Section~\ref{sec:concl} concludes the paper.

\section{Language}
\label{sec:lang}
This section defines the syntax of the source and target languages
used in the formalization of our translation.
To highlight the core ideas, we deliberately keep the languages minimal.
The actual implementation supports more primitives such as
function definitions and adopts a C-like syntax.

\subsection{Syntax of the Languages}

The source language is essentially a \textsc{While} language~\cite{10.1145/363235.363259} with arrays, except that loops are expressed using ``for'' instead of ``while''.
The target language shares the same core syntax as the source language, but it uses streams instead of arrays.
Functions are omitted without loss of generality, 
since recursive functions are typically not supported in HLS programs, and non-recursive functions can be inlined.

The set of statements and arithmetic expressions are defined as follows:
\begin{align*}
    \text{(Statements)} \ s &::= {\color{blue} x := a[e] \mid a[e]:=x} \\ 
      &\mid {\color{olive} x := a.read() \mid a.write(x)} \\
      &\mid x := e \mid s_1; s_2 \mid \kernelexp{s} \\ 
      &\mid \ifexp{x}{s_1}{s_2}\\ 
      &\mid \forexp{x}{e}{m}{n}{s}\\
   \text{(Arithmetic expressions)} \ e &::= n\mid x\mid x\ \OP\ y
\end{align*}
\noindent
The first line defines the array access operations for the source language,
and the second line defines the stream operations for the target language.
The remaining lines define the syntax common to both the source and target languages.
Here, \(x, y, ...\) are variables, and
\(\OP\) denotes integer operations such as \(+\) and \(-\);
\(a\) denotes an array variable in the source language,
or a stream variable in the target language.\footnote{We assume that \emph{all} arrays are translated into streams by our translation to simplify the exposition.
  The translation can easily be modified to apply only to \emph{some} arrays, as in our implementation.}

We briefly explain the meaning of the statements; the formal semantics is given later in Section~\ref{sec:opsem}.
A branch statement \(\ifexp{x}{s_1}{s_2}\) executes \(s_1\) if \(x\) is non-zero, and \(s_2\) otherwise.
A for-loop is written as \(\forexp{x}{e}{m}{n}{s}\), which executes the body \(s\) while \(x \ne m\), 
allowing both ascending and descending iterations to be uniformly expressed.  
The construct \(\BeginKer\cdot\EndKer\) is just a marker indicating that \( s \) is kernel code.
It does not affect program execution but plays a role in the translation, as described in Section~\ref{sec:buffer}.
We require that \(\kernelexp{s}\) cannot be nested, i.e.~\(s\) cannot have a substatement of the form \(\kernelexp{s'}\).

\begin{example}
  \label{ex:filter_origin}
  Recall the naïve filtering function in Figure~\ref{fig:intro1}.
  The kernel function and host code which calls it can be expressed in our source language as follows:
  \begin{align*}
    & \forexp{x}{0}{N}{1}{} \{ \\
    & \quad in[x] := x; \\
    & \}; \\
    & \BeginKer \\
    & \forexp{x}{0}{N-1}{1}{} \{ \\
    & \quad y_0 := in[x];\\
    & \quad y_1 := in[x+1];\\
    & \quad z_0 := y_0 + y_1;\\
    & \quad z_1 := z_0 / 2;\\
    & \quad out[x] := z_1\\
    & \} \\
    & \EndKer
  \end{align*}
  
  \noindent
  The statement of the kernel function is indicated by \(\BeginKer\) and \(\EndKer\). 
  We will use (variants of) this code as a running example throughout the paper.
  \qed
\end{example}

\subsection{Operational Semantics}
\label{sec:opsem}

This section defines the big-step operational semantics of the source language.
The reduction relation is of the form \( \sconfig R H s \eval \sstate {R'} {H'} \) meaning that if we evaluate the statement \( s \) under the \emph{state} \( \sstate R H \) we result in a state \( \sstate {R'} {H'} \).
A state is a pair of a \emph{register file} \( R \), which is a finite map from variables to integers, and a \emph{heap} \( H \), which is a finite map from \emph{addresses} of the form \( (a, m) \) to integers.
Intuitively, the array variable \( a \) represents an address of a base pointer and the integer \( m \) represents the offset from the base pointer; since we use array variables as a part of the address, different arrays never collide.
We write \( \dom(H) \) for the domain of \( H \), and \( H \set{(a, m) \mapsto n}\) for a heap that maps \( (a, m) \) to \( n \) and maps all the other addresses according to \( H \).
Notations \( \dom(R) \) and \( R\set{x \mapsto n} \) for register files are defined similarly.

\begin{figure*}
\fbox{\( \econfig{R}{e} \eval n \)}\hfill\phantom{dummy}\par
\begin{minipage}{.28\textwidth}
  \infrule[R-Const]
    {}
    {\econfig{R}{n} \eval n}
\end{minipage}
\begin{minipage}{.25\textwidth}
\infrule[R-Var]
    {R(x)=n}
    {\econfig{R}{x} \eval n}
\end{minipage}
\begin{minipage}{.42\textwidth}
\infrule[R-Op]
    {R(x)=n_1\andalso R(y)=n_2}
    {\econfig{R}{x\ \OP\ y} \eval{\llbracket\OP\rrbracket(n_1, n_2)}}
\end{minipage}
\par\bigskip
\fbox{\( \sconfig{R}{H}{s} \eval \sstate{R'}{H'} \)}\hfill\phantom{dummy}
\par
\bigskip
\begin{minipage}{.48\textwidth}
\infrule[R-ReadArray]
    {\econfig{R}{e}\eval m\andalso H(a,m) = n}
    {\sconfig{R}{H}{x:=a[e]} \eval \sstate{R\set{x\mapsto n}}{H}}
\end{minipage}%
\begin{minipage}{.5\textwidth}
\infrule[R-WriteArray]
    {\econfig{R}{e}\eval m} 
    {\sconfig{R}{H}{a[e]:=x} \eval \sstate{R}{H\set{\Addr{a}{m}\mapsto R(x)}}}
\end{minipage}
\par\bigskip
\begin{minipage}{.43\textwidth}
\infrule[R-Assign]
    {\econfig{R}{e}\eval n}
    {\sconfig{R}{H}{x:=e} \eval \sstate{R\set{x\mapsto n}}{H}}
\end{minipage}%
\begin{minipage}{.5\textwidth}
\infrule[R-Seq]
        {\sconfig{R}{H}{s_1}\eval \sstate{R''}{H''}\andalso
         \sconfig{R''}{H''}{s_2}\eval \sstate{R'}{H'}}
    {\sconfig{R}{H}{s_1;s_2} \red \sstate{R'}{H'}}
\end{minipage}

\par\bigskip
\begin{minipage}{.5\textwidth}
  \infrule[R-IfTrue]
    {R(x) \ne 0\andalso \sconfig{R}{H}{s_1}\eval \sstate{R'}{H'}}
    {\sconfig{R}{H}{\ifexp{x}{s_1}{s_2}} \eval \sstate{R'}{H'}}
\end{minipage}%
\begin{minipage}{.5\textwidth}
\infrule[R-IfFalse]
    {R(x) = 0\andalso \sconfig{R}{H}{s_2}\eval \sstate{R'}{H'}}
    {\sconfig{R}{H}{\ifexp{x}{s_1}{s_2}} \eval \sstate{R'}{H'}}
\end{minipage}

\infrule[R-ForLoop]
        {\econfig{R}{e}\eval k \andalso k \ne m \\
          \sconfig{R\set{x\mapsto k}}{H}{s'}\eval \sstate{R'}{H'} \quad
        s' \equiv {s; \forexp{x}{x+n}{m}{n}{s}}}
    {\sconfig{R}{H}{\forexp{x}{e}{m}{n}{s}}\eval \sstate{R'}{H'}}

\begin{minipage}[b]{.65\textwidth}%
 \infrule[R-ForExit]
    {\econfig{R}{e}\eval k \andalso k = m}
    {\sconfig{R}{H}{\forexp{x}{e}{m}{n}{s}}\eval \sstate{R\set{x \mapsto k}}{H}}                              %
\end{minipage}%
\begin{minipage}[b]{.35\textwidth}
\infrule[R-CallKer]
        {\sconfig{R}{H}{s}\eval \sstate{R'}{H'}}
        {\sconfig{R}{H}{\kernelexp{s}}\eval \sstate{R'}{H'}}
\end{minipage}

    \caption{Reduction rules for the source language.}
    \label{fig:lang}

\end{figure*}

\begin{figure*}
%
\infrule[R-Read]
    {\Strm(a) = n\cdot \seq{m}}
    {\sconfig{R}{\Strm}{x:=\RStr{a}} \eval \sstate{R\set{x\mapsto n}}
      {\Strm\set{a\mapsto \seq{m}}}}
\par\bigskip
\infrule[R-Write]
    {} 
    {\sconfig{R}{\Strm}{\WStr{a}{x}}
      \eval \sstate{R}{\Strm\set{a\mapsto \Strm(a)\cdot R(x)}}}
    \caption{Reduction rules for the stream operations in the target language.}
    \label{fig:lang2}
\end{figure*}

The reduction relation for the source language is defined by the rules in Figure~\ref{fig:lang}, and the rules are mostly standard.
We only explain points that may need clarification.
In rule \rn{R-Op}, \( \llbracket \OP \rrbracket \) denotes the mathematical function represented by \( \OP \).
The loop expression \(\forexp{x}{e}{m}{n}{s}\) first executes the initialization statement \( x := e \), which is executed only once, and then checks the condition \( x \ne m \).
If the condition does not hold, the for-loop terminates as expressed by \rn{R-ForExit}.
Otherwise, the statement \( s \) inside the body of the for-loop is executed, and the loop continues with \(x\) incremented by \( n \) as defined by \rn{R-ForLoop}.
We note that the variable \( x \) is not local to the for-loop. (In fact, all variables in the program are global variables.)
The rule \rn{R-CallKer} just says that \(\BeginKer\) and \(\EndKer\) do not affect the reduction.

The reduction rules for the stream operations in the target language are given in Figure~\ref{fig:lang2}.
Here the state takes a \emph{stream pool} \( S \) instead of a heap.
A stream pool is a finite map from array variables to (possibly empty) integer sequences.
We write \( \seq m \) for an integer sequence and \( n \cdot \seq m \) (resp.~\( \seq m \cdot n \)) for the sequence obtained by adding \( n \) to the head (resp.~end) of \( \seq m\).
The rule \rn{R-Read} pops an integer from the head of a stream, and assigns it to \( x \) (if the stream is empty, the reduction is stuck), and \rn{R-Write} adds the value assigned to \( x \) to the end of the sequence.
The remaining rules are identical to those of the source language, and are therefore omitted.

\section{Formalization of Translation}
\label{sec:buffer}

This section defines the translation for inserting buffering commands and replacing array accesses with stream operations.
The translation relation is defined by using a kind of relational Hoare logic.
The source and target languages for the translation are the language defined in Section~\ref{sec:lang}.

\subsection{Translation Relation}
\label{sec:buffer-judgment}

The set of types, ranged over by \(\ty\), is given by:
\begin{align*}
  \ty&::=\INT \mid \BUF \mid \RARR \mid \WARR
\end{align*}
Here, \(\INT\) is the type of integer expressions, and \(\BUF\) is the type of buffer variables;
a variable of type \(\BUF\) is also used to store an integer, but only occurs in the target of the translation.
The types \(\RARR\) and \(\WARR\) respectively describe read- and write-only arrays.

The translation relation is of the form \(\bj{\ATE}{\BTE}{s}{\BTE'}{t}\), where
(i) \(\ATE\) is a type environment of the form \(x_1\COL\ty_1,\ldots,x_n\COL\ty_n\);
(ii) \(s\) and \(t\) are source and target programs respectively;
and
(iii) \(\BTE\) and \(\BTE'\) are respectively the pre- and post-conditions;
i.e., \(\BTE\) (\(\BTE'\), resp.) describes a (relational) property that holds for 
the states of the source and target programs before (after, resp.) the executions of \(s\) and \(t\).
We will also introduce an auxiliary relation of the form \(\bjs{\ATE}{\BTE}{\BTE'}{t}\), which 
means that \(t\) may be inserted in the place where \(\BTE\) holds; \(\bullet\) may be regarded as
a skip command.

For example, we have:
\begin{align*}
  \bj{x\COL\INT,a\COL\RARR,b\COL\BUF}{b=a[1]}{x:=a[1]}{x=a[1]}{x:=b}.
\end{align*}
It means that if \(b=a[1]\) holds, i.e., the value of \(a[1]\) is stored in \(b\)
(which serves as a buffer),
then we can replace \(x:=a[1]\) in the source program with \(x:=b\), to avoid reading \(a[1]\) again.
Specifying relational properties of programs using pre- and post-conditions is in the style of relational Hoare logic.
In the above example, the pre-condition expresses a relational property:
\(b\) refers to the value of \(b\) in the target program, while \(a[1]\) refers to the value of
\(a[1]\) in the source program.

Each array \(a\) in the source program \(s\) is converted to a stream (of the same name \(a\)) in the target program \(t\).
For each array variable \(a\), \(\BTE\) may refer to the special variable \(\iset{a}\),
which denotes a \emph{sequence} of indices of \(a\),
indicating the order in which elements are stored in the stream \(a\) in the target program.
For instance, \(\iset{a}=1\cdot 3\cdot 5\)
means that the values of \(a[1]\), \(a[3]\) and \(a[5]\) are stored in this order in the stream \(a\).
As a concrete example, we have:
\begin{align*}
  \bj{\ATE}{\iset{a}=1\cdot J}{x:=a[1]}{\iset{a}=J\land x=b=a[1]}{b:=\RStr{a};x:=b}
\end{align*}
where \(\ATE = x\COL\INT,a\COL\RARR,b\COL\BUF\).
The judgment means that if \(\iset{a}=1\cdot J\) holds,
then the command \(x:=a[1]\) may be replaced with \(b:=\RStr{a};x:=b\),
and \(\iset{a}=J\land x=b=a[1]\) hold afterwards.
The pre-condition \(\iset{a}=1\cdot J\) ensures that the value of \(a[1]\) is stored in the head of the stream in the target
program; thus, the array read \(a[1]\) can safely be replaced with the stream read \(\RStr{a}\), and \(\iset{a}\) is updated to
\(J\) afterwards, as expressed by the post-condition \(\iset{a}=J\). 
The other post-condition \(x=b=a[1]\) states that \(x\) and \(b\) hold the value of \(a[1]\) in the source program.
\Revision{Note that a sequence of indices refers to a finite sequence, whose elements may include variables or expressions.}


\subsection{Translation Rules}

\begin{figure*}[tbp]
\typicallabel{Tr-ReadMem}

\vspace*{1ex}
\infrule[Tr-ReadMem]
        {\ATE(x)=\INT\andalso \ATE(a)=\RARR\\ \ATE\p e:\INT\andalso \ATE(b)=\BUF}
{\bj{\ATE}{b=a[e]\land [b/x]\BTE}{x := a[e]}{\BTE}{x := b}}

\vspace*{1ex}
\infrule[Tr-WriteMem]
{\ATE(x)=\INT\andalso \ATE(a)=\WARR\\ \ATE\p e:\INT\andalso \ATE(b)=\BUF}
{\bj{\ATE}{e\not\in \iset{a}\land [x/b, \updateAr{a}{e}{x}/a]\BTE}{a[e] := x}{\BTE}{b := x}}

\vspace*{1ex}
\infrule[Tr-Assign]
{\ATE(x)=\INT\andalso \ATE\p e:\INT}
{\bj{\ATE}{[e/x]\BTE}{x := e}{\BTE}{x := e}}


\vspace*{1ex}
\infrule[Tr-Seq]
{\bj{\ATE}{\BTE}{s_1}{\BTE''}{\texp_1} \andalso
 \bj{\ATE}{\BTE''}{s_2}{\BTE'}{\texp_2}}
{\bj{\ATE}{\BTE}{s_1;s_2}{\BTE'}{\texp_1;\texp_2}}

\vspace*{1ex}
\infrule[Tr-If]
{\ATE(x)=\INT \andalso \bj{\ATE}{\BTE\land x\ne 0}{s_1}{\BTE'}{\texp_1} \\ 
 \bj{\ATE}{\BTE\land x= 0}{s_2}{\BTE'}{\texp_2}}
{\bj{\ATE}{\BTE}{\ifexp{x}{s_1}{s_2}}{\BTE'}{\ifexp{x}{\texp_1}{\texp_2}}}

\vspace*{1ex}
\infrule[Tr-For]
{  \bj{\ATE}{\BTE \land x \ne m}{s}{[x+n/x]\BTE}{\texp}}
{\ATE \p \{[e/x]\BTE\}\forexp{x}{e}{m}{n}s \\
 \To \forexp{x}{e}{m}{n}{\texp}\{\BTE\land x=m\}}



\vspace*{1ex}
\infrule[Tr-InsertL]
{\bjs{\ATE}{\BTE}{\BTE''}{\texp_0}\andalso  \bj{\ATE}{\BTE''}{s}{\BTE'}{\texp}}
{ \bj{\ATE}{\BTE}{s}{\BTE'}{\texp_0;\texp}}

\vspace*{1ex}
\infrule[Tr-InsertR]
{\bj{\ATE}{\BTE}{s}{\BTE''}{\texp}\andalso \bjs{\ATE}{\BTE''}{\BTE'}{\texp_0}}
{ \bj{\ATE}{\BTE}{s}{\BTE'}{\texp;\texp_0}}

\vspace*{1ex}
\infrule[Tr-InsRBuf]
        {\ATE(a)=\RARR\andalso \ATE(b)=\BUF\\
         \BTE =  [a[\hd(\iset{a})]/b,\tl(\iset{a})/\iset{a}]\BTE'}
        { \bjs{\ATE}{\BTE}{\BTE'}
          {b:=\RStr{a}}}
\vspace*{1ex}
\infrule[Tr-InsWBuf]{\ATE(a)=\WARR\andalso  \ATE(b)=\BUF}
        {\bjs{\ATE}{n \not\in \iset{a}\land
                a[n]=b\land [ \iset{a}\cdot n/\iset{a}]\BTE}{\BTE}{\WStr{a}{b}}}
        
%

\vspace*{1ex}
\infrule[Tr-InsMove]
{\ATE(x)=\BUF\andalso \ATE(y)\in \set{\INT,\BUF}}
{ \bjs{\ATE}{[y/x]\BTE}{\BTE}{x:=y}}

\vspace*{1ex}
\infrule[Tr-Conseq]
        {\models \BTE\Imp \BTE_1\andalso \bj{\ATE}{\BTE_1}{s}{\BTE_1'}{\texp}\andalso \models \BTE_1'\Imp \BTE'}
        {\bj{\ATE}{\BTE}{s}{\BTE'}{\texp}}

\vspace*{1ex}
\infrule[Tr-InsConseq]
        {\models \BTE\Imp \BTE_1\andalso \bjs{\ATE}{\BTE_1}{\BTE_1'}{\texp}\andalso \models \BTE_1'\Imp \BTE'}
        {\bjs{\ATE}{\BTE}{\BTE'}{\texp}}

\vspace*{1ex}
\infrule[Tr-Kernel]
{\bj{\fliparray{\ATE}}{\BTE}{s}{\BTE'}{\texp}}
{\bj{\ATE}{\BTE}{\kernelexp{s}}{\BTE'}{\kernelexp{\texp}}}

\caption{Translation rules. \(\ATE\p e:\INT\) means that all the variables in \(e\) have type \(\INT\).}
\label{fig:fig-buf1}

\end{figure*}

The translation relation \(\bj{\ATE}{\BTE}{s}{\BTE'}{t}\) is defined by the rules in Figure~\ref{fig:fig-buf1}.
In the figure, \([e_1/x_1,\ldots,e_n/x_n]\) denotes the simultaneous substitutions of \(e_1,\ldots,e_n\) for \(x_1,\ldots,x_n\).
The expression \(a\set{e\mapsto x}\) denotes the array obtained from \(a\) by replacing the value at index \(e\) with \(x\);
the syntax of the assertion language for describing \(\BTE\) 
is\CameraReady{
given in the full version~\cite{tanaka2026relationalhoarelogichighlevel}.
}{
given in Appendix~\ref{appx:correctness-stream}.
}  

Most of the rules in the figure are analogous to the corresponding (partial correctness) rules of Hoare logic.
In the rule \rn{Tr-ReadMem}, a read from the array \(a\) is translated to a read from the buffer variable \(b\),
on the assumption that the value of \(a[e]\) has already been copied to \(b\); that is achieved by the rule \rn{Tr-InsRBuf} explained later. 
Thus, when \(a[e]\) is accessed for the first time, \(x:=a[e]\) is translated to \(b:=\RStr{a}; x:=b\), which can be optimized to
\(x:=\RStr{a}\) if the value of \(a[e]\) is not used later; this redundancy is for the sake of simplicity, and in the actual implementation,
\(x:=a[e]\) is just transformed to \(x:=\RStr{a}\) if \(a[e]\) is not accessed more than once.
The pre-condition \(b=a[e]\) asserts that \(b\) indeed holds the value of \(a[e]\). We additionally require
\([b/x]\BTE\) as a pre-condition to ensure that \(\BTE\) holds after the execution of \(x:=a[e]\) and \(x:=b\).
This is analogous to the rule \(\set{[b/x]\BTE}x:=b\set{\BTE}\) for assignment in the standard Hoare logic.
In our case, \(x\) is updated to \(a[e]\) and \(b\) respectively in the source and target programs, but requiring
\([b/x]\BTE\) suffices since the other condition \(b=a[e]\) ensures that \([a[e]/x]\BTE\) also holds.

Similarly, in the rule \rn{Tr-WriteMem}, a write to the array \(a\) is translated to a write to the buffer variable \(b\).
The code for writing the contents of the buffer to the stream is then inserted by the rule \rn{Tr-InsWBuf} explained later.
Thus, if \(a[e]\) is updated just once in the source program, \(a[e]:=x\) is translated to \(b:=x; \WStr{a}{b}\), which can be optimized
to \(\WStr{a}{x}\).
This is again for the sake of simplicity; \(a[e]:=x\) is directly translated to \(\WStr{a}{x}\) in the actual implementation in such a case.
The pre-condition \(e\not\in \iset{a}\) checks that \(a[e]\) has not been written before. 
The part \([x/b, \updateAr{a}{e}{x}/a]\BTE\) corresponds to the precondition for
the two successive assignments \(a[e]:=x;b:=x\) in the standard Hoare logic.
Here, the substitution can be applied in parallel because
the two assignments do not interfere; note that (i) the buffer variable \(b\) may occur only in the target program,
and (ii) the array variable \(a\) in pre-/post-conditions refers to that of the source program.

The rules \rn{Tr-Assign}, \rn{Tr-Seq}, and \rn{Tr-If} are just relational versions of the corresponding rules
for Hoare logic. The rule \rn{Tr-For} is also essentially a relational version of the rule for loops in Hoare logic,
except that we chose for-statements as primitives instead of while-statements.
With while-statements as primitives, the rule \rn{Tr-For} would be derived from the following rule:
\infrule[Tr-While]
  {\bj{\ATE}{\BTE \land \mathit{cond}}{s}{\BTE}{t}}
  {\bj{\ATE}{\BTE}{\whileexp{\mathit{cond}}{s}}{\BTE\land \neg \mathit{cond}}{\whileexp{\mathit{cond}}{t}}}
Note that, using a while-statement, \(\forexp{x}{e}{m}{n}s\) can be expressed as
\(
x := e; \mathbf{while}\ x \ne m\ \mathbf{do}\ (s; x:=x+n)
\).

The rules \rn{Tr-InsertL} and \rn{Tr-InsertR} are for inserting buffering commands.
Here, we use the auxiliary relation \(\bjs{\ATE}{\BTE}{\BTE'}{t}\) (where \(\bullet\) can be considered a skip command)
for inserting \(u\) in the target program.
The rule \rn{Tr-InsRBuf} is for copying a value from a stream to a buffer. 
The substitution \([\tl(\iset{a})/\iset{a}]\) reflects the 
update (\(\iset{a} := \tl(\iset{a})\)) of the sequence of array indices stored in the stream.
Here, \(\hd(J)\) denotes the first element of the sequence \(J\), \(\tl(J)\) denotes the sequence obtained by removing the first element from \(J\).
In the rule \rn{Tr-InsWBuf} for copying a value from a buffer to a stream,
the pre-condition \(n \not\in \iset{a}\) checks that \(a[n]\) has not been written before. The substitution
\([(\iset{a}\cdot n)/\iset{a}]\) reflects the fact that the value of \(a[n]\) is added at the end of the stream.
Here, \(J_1\cdot J_2\) denotes the concatenation of two sequences \(J_1\) and \(J_2\).

The rule \rn{Tr-InsMove} is for moving the value of a buffer or integer variable to another buffer variable;
that is needed, for example, for inserting \(\texttt{b0 = b1}\) in Figure~\ref{fig:intro1} given in Section~\ref{sec:intro}.
The rules \rn{Tr-Conseq} and \rn{Tr-InsConseq} are for adjusting pre-/post-conditions, corresponding to the consequence rule of
the standard Hoare logic. 
The rule \rn{Tr-Kernel} handles the swapping of read-only and write-only array types when entering a kernel function.
The operation \(\fliparray{\ATE}\) on type environments flips between read/write-only array types, as defined by:
\begin{align*}
& \fliparray{(a_1\COL\ty_1,\ldots,a_n\COL\ty_n)}=a_1\COL \flip{\ty_1},\ldots,a_n\COL\flip{\ty_n}\qquad
 {\ty}^{\sharp} = \left\{\begin{array}{ll}
  \RARR & \text{if } \ty=\WARR \\
  \WARR & \text{if } \ty=\RARR \\
  \ty & \text{otherwise}
  \end{array}\right.
\end{align*}
This flipping is necessary because arrays read (resp. written) inside the kernel are written (resp. read) by the host code.

\begin{example}
  Recall the following judgment considered in Section~\ref{sec:buffer-judgment}.
  \[ \bj{\ATE}{\iset{a}=1\cdot J}{x:=a[1]}{\iset{a}=J\land x=b=a[1]}{b:=\RStr{a};x:=b}.\]
  It is obtained from the following two judgments by using \rn{Tr-InsertL}:
  \begin{align*}
  & \bjs{\ATE}{\iset{a}=1\cdot J}{\iset{a}=J\land b=a[1]}{b:=\RStr{a}}\\
  & \bj{\ATE}{\iset{a}=J\land b=a[1]}{x:=a[1]}{\iset{a}=J\land x=b=a[1]}{x:=b}.
  \end{align*}
  The first judgment is obtained as follows.
  \[
  \infer[\rn{\small (Tr-InsConseq)}]{\bjs{\ATE}{\iset{a}=1\cdot J}{\BTE}{b:=\RStr{a}}}
        {\infer[\rn{\small (Tr-InsRBuf)}]{\bjs{\ATE}{\hd(\iset{a})=1 \land 
            [a[1]/b, \tl(\iset{a})/\iset{a}]\BTE}{\BTE}{b:=\RStr{a}}}{}}
  \]
  Here, \(\BTE \equiv \iset{a}=J\land b=a[1]\).
  Note that
        \(\iset{a}=1\cdot J\) implies
        \(\hd(\iset{a})=1\land[a[1]/b, \tl(\iset{a})/\iset{a}]\BTE\), i.e., \(\hd(\iset{a})=1\land\tl(\iset{a})=J\land a[1]=a[1]\).
  The second judgment can be obtained by using
  \rn{Tr-ReadMem} and \rn{Tr-InsConseq}. \qed
\end{example}

\begin{figure*}[tbp]
      \begin{tikzpicture}
        \fill[teal!5,rounded corners=5pt] (-6.4, -7.6) rectangle (-1.1, 7.7);
        \fill[orange!5,rounded corners=5pt] (0.4, -7.6) rectangle (5.7, 7.7);
        \node at (0, 0) [align=center] {
          \begin{minipage}{\textwidth}
            \begin{flalign*}
              \bte{//  $\ATE; \iset{\In}=\iset{\Out}=\epsilon$}\\
              & \forexp{x}{0}{N}{1}{}\{ && \forexp{x}{0}{N}{1}{}\{\\
              \bte{//  $\ATE; \iset{\In}= 0\cdots (x-1)\land \iset{\Out}=\epsilon\land x \ne N$}\\
              & \quad in[x] := x; && \quad b := x;\\
              \bte{//  $\ATE; \iset{\In}= 0\cdots (x-1)\land \iset{\Out}=\epsilon \land b=in[x] \land x \ne N$}\\
              &  && \quad \WStr{in}{b};\\
              \bte{//  $\ATE; \iset{\In}= 0\cdots x \land \iset{\Out}=\epsilon\land x \ne N$}\\
              & \}; && \};\\
              \bte{//  $\ATE; \iset{\In}= 0\cdots (N-1)\land \iset{\Out}=\epsilon$}\\
              &  \BeginKer && \BeginKer \\
              \bte{//  $\ATE^{\sharp}; \iset{\In}= 0\cdots (N-1)\land \iset{\Out}=\epsilon$}\\
              &  && b_0 := \RStr{in}\\
              \bte{//  $\ATE^{\sharp}; \iset{\In}= 1\cdots (N-1)\land \iset{\Out}=\epsilon\land b_0=in[0]$}\\
              & \forexp{x}{0}{N-1}{1}{}\{  && \forexp{x}{0}{N-1}{1}{}\{ \hspace{1.5em} \\ 
              \bte{//  $\ATE^{\sharp}; \iset{\In}= (x+1)\cdots (N-1)\land \iset{\Out}=0\cdots (x-1)\land b_0=in[x] \land x \ne N-1$}\\
              & \quad y_0 := in[x];&& \quad y_0 := b_0;\\
              \bte{//  $\ATE^{\sharp}; \iset{\In}= (x+1)\cdots (N-1)\land \iset{\Out}=0\cdots (x-1)\land x \ne N-1$}\\
              & \quad && \quad b_0 := \RStr{in}\\
              \bte{//  $\ATE^{\sharp}; \iset{\In}= (x+2)\cdots (N-1)\land \iset{\Out}=0\cdots (x-1)\land b_0=in[x+1] \land x \ne N-1$}\\
              & \quad y_1 := in[x+1]; && \quad y_1 := b_0;\\
              & \quad z_0 := y_0 + y_1;    && \quad z_0 := y_0 + y_1;\\
              & \quad z_1 := z_0 / 2;  && \quad z_1 := z_0 / 2;\\
              \bte{//  $\ATE^{\sharp}; \iset{\In}= (x+2)\cdots (N-1)\land \iset{\Out}=0\cdots (x-1)\land b_0=in[x+1] \land x \ne N-1$}\\
              & \quad out[x] := z_1 && \quad b_1 := z_1;\\
              \bte{//  $\ATE^{\sharp}; \iset{\In}= (x+2)\cdots (N-1)\land \iset{\Out}=0\cdots (x-1)\land b_0=in[x+1]$}\\
              && \bte{ \hspace{2.5em} $ {} \land b_1=out[x] \land x \ne N-1$}\\ 
              & \quad  && \quad \WStr{out}{b_1}\\
              \bte{//  $\ATE^{\sharp}; \iset{\In}= (x+2)\cdots (N-1)\land \iset{\Out}=0\cdots x \land b_0=in[x+1] \land x \ne N-1$}\\
              &     \} && \}\\
              \bte{//  $\ATE^{\sharp}; \iset{\In}= \epsilon\land \iset{\Out}=0\cdots (N-2)$}\\
              &  \EndKer && \EndKer \\
              \bte{//  $\ATE; \iset{\In}= \epsilon\land \iset{\Out}=0\cdots (N-2)$}\\
            \end{flalign*}
          \end{minipage}
        };
      \end{tikzpicture}
    \caption{Translation for the running example. The left side shows the original program, and the right side shows the transformed program. 
              Comments indicate \(\BTE\) at each program point. 
              \Revision{(Here, for simplicity, we omit conditions that are not relevant to buffering or stream replacement.)}}
    \label{fig:ex-stream}
\end{figure*}

\begin{example}
  \label{ex:buffer}
  Recall the program in Example~\ref{ex:filter_origin}:
    \begin{align*}
      & \forexp{x}{0}{N}{1}{}  \\
      & \quad \{ in[x] := x \}; \\
      & \forexp{x}{0}{N-1}{1}{}  \\
      & \quad \{ y_0 := in[x]; 
      y_1 := in[x+1]; 
      z_0 := y_0 + y_1; 
      z_1 := z_0 / 2; 
      out[x] := z_1 \}.
    \end{align*}
    It is transformed as shown in Figure~\ref{fig:ex-stream}.
    For readability, instead of showing a derivation tree based on the transformation rules,
    the figure shows the source and target programs side-by-side, with \(\BTE\) for each program point
    shown as a comment.
    The type environment \(\ATE\) used in the translation is:
    \begin{align*}
      \In\COL \RARR, \Out\COL\WARR, y_0\COL\INT, y_1\COL\INT, z_0\COL\INT, z_1\COL\INT, b_0\COL\BUF, b_1\COL\BUF.
    \end{align*}
    For example, the three lines before the second for-loop represent the judgment:
    \begin{align*}
      \ATE^{\sharp} \p \{\iset{\In}= 0 & \cdots (N-1) \land \iset{\Out}=\emptyset\} \bullet \To \\
      & b_0:=\RStr{\In} \{\iset{\In}=1 \cdots (N-1)\land \iset{\Out}=\emptyset\land b_0=\In[0]\}.
    \end{align*}
    It is obtained by using the rules \rn{Tr-InsRBuf} and \rn{Tr-InsConseq} as in the previous example.
    \qed
\end{example}

Note that the order of array access does not matter
as long as the producer and consumer follow the same access sequence.
For example, consider a filter program that reads only the even indices of the input array.
In this case, the access sequence in the pre-/post-condition becomes \(\iset{\In} = 0\cdot 2\cdots (2N-2)\).
While Vitis HLS cannot apply burst data transfer optimization to such non-consecutive patterns,
our translation treats them in the same way as for a program with consecutive accesses.
Similarly, the programs in Figure~\ref{fig:intro-pitfall} can be translated identically, 
regardless of whether the access order is forward or reversed.

\newcommand\Inp[1]{\mathit{in}_{#1}}
\newcommand\IF{\mathbf{if}}
\newcommand\ELSE{\mathbf{else}}
\newcommand\btes[1]{\bte{\footnotesize #1}}
\newcommand\emptyseq{\epsilon}
\newcommand\Dots{..}

\begin{example}
  \label{ex:merge}
    Figure~\ref{fig:ex-stream-merge} shows a translation
    for a program that merges two arrays \(\Inp1\) and \(\Inp2\) into \(\Out\) in the increasing order,
    assuming that the elements of \(\Inp1\) and \(\Inp2\) are sorted in the increasing order. For readability, the example uses the ordinary if-expressions, which are easy to support.
    As in Example~\ref{ex:buffer}, the source and target of the stream translation
    are shown side-by-side, and comments denote \(\BTE\) used at each program point.
    We write \(i\Dots j\) for the sequence consisting of integers from \(i\) to \(j\); it denotes the empty sequence
    \(\emptyseq\) when \(i>j\).
    The type environment \(\ATE\) is \(\Inp1\COL\RARR, \Inp2\COL\RARR, \Out\COL\WARR, i\COL\INT,j\COL\INT,k\COL\INT,x\COL\INT,
    N\COL\INT\). In this example, the loop invariant \(i+j=k<2N\land i\le N\land j\le N\) is required for the translation.
    In fact, in order to ensure that \(i\in \iset{\Inp1}\) in the first conditional branch, the condition \(i< N\) is required,
    which can be obtained by the entailment
    \(i+j<2N\land j=N \Imp i<  N\).
    Here, \(i+j<2N\) and \(j=N\) respectively come from the loop invariant and the branching condition.
    Thanks to the use of relational Hoare logic, we can correctly reason about such invariants.
    \qed
\end{example}

\begin{figure}
  \begin{tikzpicture}
    \fill[teal!5,rounded corners=5pt] (-6.3, -7.6) rectangle (-1.1, 7.7);
    \fill[orange!5,rounded corners=5pt] (0.5, -7.6) rectangle (5.7, 7.7);
    \node at (0, 0) [align=center] {
      \begin{minipage}{\textwidth}
        \begin{flalign*}
          \btes{//  $\iset{\Inp1}=\iset{\Inp2}=0\Dots(N-1)\land \iset{\Out}=\emptyseq\land i=j=0$}\\
          & \forexp{k}{0}{N*2}{1}{}\{ && \forexp{k}{0}{N*2}{1}{}\{ \hspace{1.5em}  \\
          \btes{//  $\iset{\Inp1}=i\Dots(N-1)\land \iset{\Inp2}=j\Dots(N-1)\land \iset{\Out}=0\Dots(k-1)\land i+j=k<2N$}\\
          && \bte{ \hspace{7em} $ {} \land i\le N\land j\le N$}\\
          &\ \ \IF\ j=N\ \{ &&\ \ \IF\ j=N\ \{\\
          \btes{//  $\iset{\Inp1}=i\Dots(N-1)\land \iset{\Inp2}=\emptyseq\land \iset{\Out}=0\Dots(k-1)\land i+N=k<2N\land i< N\land j=N$}\\
          &\quad x:=\Inp1[i]; \Out[k]:= x; &&\quad x := \RStr{\Inp1}; \WStr{\Out}{x};\\
          & \quad i:=i+1 && \quad i:=i+1\\
          \btes{//  $\iset{\Inp1}=i\Dots(N-1)\land \iset{\Inp2}=\emptyseq\land \iset{\Out}=0\Dots k\land i-1+N=k<2N\land i\le N\land j=N$}\\
          & \ \ \}\ \ELSE\ \IF\ i=N\ \{&&\ \ \}\ \ELSE\ \IF\ i=N\ \{\\
          \btes{//  $\iset{\Inp1}=\emptyseq\land \iset{\Inp2}=j\Dots(N-1)\land \iset{\Out}=0\Dots(k-1)\land N+j=k<2N\land i= N\land j < N$}\\
          &\quad x:=\Inp2[i]; \Out[k]:= x; &&\quad x := \RStr{\Inp2}; \WStr{\Out}{x};\\
          & \quad j:=j+1 &&\quad j:=j+1\\
          \btes{//  $\iset{\Inp1}=\emptyseq\land \iset{\Inp2}=j\Dots(N-1)\land \iset{\Out}=0\Dots k\land N+j-1=k<2N\land i= N\land j\le N$}\\
          &\ \  \}\ \ELSE\ \IF\ \Inp1[i]<\Inp2[j]\ \{ &&\ \  \}\ \ELSE\ \IF\ \Inp1[i]<\Inp2[j]\ \{ \\
          \btes{//  $\iset{\Inp1}=i\Dots(N-1)\land \iset{\Inp2}=j\Dots(N-1)\land \iset{\Out}=0\Dots(k-1)\land i+j=k<2N$}\\
          && \bte{ \hspace{7em} $ {} \land i <  N\land j < N$}\\
          &\quad x:=\Inp1[i]; \Out[k]:= x; &&\quad x := \RStr{\Inp1}; \WStr{\Out}{x}; \\
          &\quad i:=i+1 &&\quad i:=i+1\\
          \btes{//  $\iset{\Inp1}=i\Dots(N-1)\land \iset{\Inp2}=j\Dots(N-1)\land \iset{\Out}=0\Dots k\land i-1+j=k<2N$}\\
          && \bte{ \hspace{7em} $ {} \land i\le N\land j < N$}\\
          &\ \  \}\ \ELSE\ \{&&\ \ \}\ \ELSE\ \{\\
          \btes{//  $\iset{\Inp1}=i\Dots(N-1)\land \iset{\Inp2}=j\Dots(N-1)\land \iset{\Out}=0\Dots(k-1)\land i+j=k<2N$}\\
          && \bte{ \hspace{7em} $ {} \land i\le N\land j\le N$}\\
          &\quad x:=\Inp2[i]; \Out[k]:= x; &&\quad x := \RStr{\Inp2}; \WStr{\Out}{x};\\
          &\quad j:=j+1 &&\quad j:=j+1\\
          \btes{//  $\iset{\Inp1}=i\Dots(N-1)\land \iset{\Inp2}=j\Dots(N-1)\land \iset{\Out}=0\Dots k\land i+j-1=k<2N$}\\
          && \bte{ \hspace{7em} $ {} \land i < N\land j \le N$}\\
          &\ \  \} &&\ \  \} \\
          \btes{//  $\iset{\Inp1}=i\Dots(N-1)\land \iset{\Inp2}=j\Dots(N-1)\land \iset{\Out}=0\Dots k\land i+j=k+1\le 2N$}\\
          && \bte{ \hspace{7em} $ {} \land i\le N\land j \le N$}\\
          & \} && \} \\
          \btes{//  $\iset{\Inp1}=\emptyseq\land \iset{\Inp2}=\emptyseq\land \iset{\Out}=0\Dots (2N)\land k=2N\land i= j=N$}\\
        \end{flalign*}
      \end{minipage}
    };
  \end{tikzpicture}
\caption{Translation for a merge program.}
\label{fig:ex-stream-merge}
\end{figure}

\subsection{Correctness of the Translation}
To show that the semantics of the program is preserved by the translation,
we define a relation \(\relRH{R}{H}{R'}{\Strm'}{\ATE}{\BTE}\) between
the runtime states of programs before and after the translation as follows.
\infrule{R(x)=R'(x)\mbox{ for each $x\COL\INT\in \ATE$}\\
  \Iset{H}{a} \supseteq \setof{\Iseq(a)}
  \mbox{ for each $a$ such that $\ATE(a)\in\set{\RARR,\WARR}$}\\
  H(a,\Iseq(a)[i]) = \Strm'(a)[i]
  \mbox{ for each $a$ and $i$ such that} \hspace{2cm} \\ 
  \hspace{4cm}\mbox{$\ATE(a)\in\set{\RARR,\WARR}$, $0\le i < \length{\Iseq(a)}$}\\
  R',H, \Iseq\models \BTE
}
{\relRH{R}{H}{R'}{\Strm'}{\ATE}{\BTE}}
Here, \(\Iseq(a)\) denotes a sequence of indices \sk{Should we emphasize that the elements of the sequence are pairwise distinct?} of array \(a\) and \( \Iseq(a)[i]\) denotes the \( i \)-th element of the sequence; the sequence describes
how the values of an array \(a\) of the source program are stored in the stream \(a\) of
the target program.
For example, if \(\Strm'(a)=2\cdot 3\) and \(\Iseq(a)=1\cdot 0\),
then \(H(a,1)=2\) and \(H(a,0)=3\).
To avoid undefined accesses to the heap, we require that the elements of \( \Iseq(a)\) are included in \( \Iset H a \defeq \{ m \mid (a, m) \in \dom(H)\} \).
The last premise \(R',H, \Iseq\models \BTE\) means that
\(\BTE\) holds under the interpretation which maps (i) each integer and buffer variable \(x\) to
\(R'(x)\), (ii) each array \(a\) to \(\set{m\mapsto H(a,m)}\), and
(iii) each index sequence variable \(\iset{a}\) to \(\Iseq(a)\).
 Note here that array variables are interpreted by using the state of the program
\emph{before} the translation, while integer and buffer variables and \( \iset a \) are interpreted by using the state of the
program \emph{after} the translation.

The following theorem on the correctness of the translation states
that the evaluation of source and target programs can be simulated by each other.
\begin{restatable}[Correctness of translation]{theorem}{correctnessStream}
  \label{lem:str-correctness}
  Suppose \(\sj{\ATE}{\BTE}{s}{\BTE'}{t}\)
  and \(\relRH{R}{H}{R'}{\Strm'}{\ATE}{\BTE}\).
  \begin{enumerate}
  \item If \(\sconfig{R}{H}{s}\eval \sstate{R_1}{H_1}\), then
    there exist \(R_1'\) and \(\Strm_1'\) such that
    \(\sconfig{R'}{\Strm'}{\iexp}\eval \sstate{R'_1}{\Strm'_1}\) and
    \(\relRH{R_1}{H_1}{R'_1}{\Strm'_1}{\ATE}{\BTE'}\).
  \item If \(\sconfig{R'}{\Strm'}{t}\eval \sstate{R'_1}{\Strm'_1}\), then
    there exist \(R_1\) and \(H_1\) such that
    \(\sconfig{R}{H}{s}\eval \sstate{R_1}{H_1}\) and
    \(\relRH{R_1}{H_1}{R'_1}{\Strm'_1}{\ATE}{\BTE'}\).
  \end{enumerate}
\end{restatable}
\begin{proof}
  \CameraReady{
    Given in the full version~\cite{tanaka2026relationalhoarelogichighlevel}.
  }{
    Given in Appendix~\ref{appx:correctness-stream}.
  }  
  \qed
\end{proof}
Since \(\relRH{R_1}{H_1}{R'_1}{\Strm'_1}{\ATE}{\BTE'}\) requires that \( R_1(x) = R'_1(x) \) for each integer variable \( x \), the above theorem says that the results of the programs before and after the translation are the same.
Furthermore, the mutual simulation ensures that the translated program never gets stuck nor diverges if and only if the source program does not.

\section{Automation of Translation}
\label{sec:auto}

Below we describe how our translation is automated; experimental results are reported
in Section~\ref{sec:experiment}.

\subsection{Restriction on Pre-conditions and Post-conditions}
\label{sec:auto-restriction}
For the purpose of automating the translation, we have restricted the shape of \(\BTE\)
to the form
\[
  \iset{a_1}=[\mathit{low}_1,\mathit{high}_1; s_1]\land \cdots \land
  \iset{a_k}=[\mathit{low}_k,\mathit{high}_k; s_k]\land \varphi,
\]
where \([\mathit{low}_i,\mathit{high}_i; s_i]\) denotes the sequence of integers from
\(\mathit{low}_i\) to \(\mathit{high}_i\) with gap \(s_i\) 
(e.g., \([0,6;2]\) denotes the sequence \(0\cdot 2\cdot 4\cdot 6\)),
and \(\varphi\) consists of only (i) buffer information \(b=a[e]\) and (ii) conditions on loop variables 
(that can automatically be extracted from for-loops; this is the reason why we preferred for-loops over while-loops).
The bounds \(\mathit{low}_i\) and \(\mathit{high}_i\) can be linear expressions on free variables
(whose coefficients are automatically inferred as described in Section~\ref{sec:auto-vcg}).

\Revision{
\paragraph{Limitations}
As described above, our current implementation imposes restrictions on index access patterns, although our formalization does not have such restrictions.
At present, the tool requires that, across loop iterations, the range of accessed indices
shifts by a constant, as in sliding-window or forward/backward scans.
Consequently, it does not support unconventional index orders such as column-major access
(e.g., \texttt{A[j][i]} instead of the usual \texttt{A[i][j]}) in matrix computations.
In future work, we plan to relax this restriction by fully exploiting the power of the relational Hoare logic, 
possibly allowing users to specify loop invariants.}

\subsection{Buffer Insertion}
Our tool first inserts buffering commands as follows.
For each loop and each array accessed in the loop, the tool collects the set of indices (as symbolic expressions)
accessed during each iteration. For example, in the source program of Figure~\ref{fig:ex-stream},
the sets of indices for arrays \(\In\) and \(\Out\) are \(\set{x,x+1}\) and \(\set{x}\) respectively.
The tool orders the indices symbolically, and decides that all indices except the largest one should be read from
buffers. On line 13 of Figure~\ref{fig:ex-stream}, the read from \(\In[x]\) is thus replaced
with a read from buffer \(b_0\), and on line 14, the value read from the corresponding stream is stored in the buffer
for use in the next iteration. In contrast, buffering is unnecessary for \(\Out\) since only a single index \(x\) of \(\Out\)
is accessed. (As already explained,
the replacement of \(\Out[x]:=z_1\) with \(b_1:=z_1; \Out.write(b_1)\) is for the sake of simplicity of
translation rules. In the actual implementation, this is replaced directly with \(\Out.write(z_1)\).)
We also note that, for two-dimensional arrays,
the tool introduces line buffers and applies the same reasoning as above for the first index of each two-dimensional array.

\subsection{Verification Condition Generation}
\label{sec:auto-vcg}
Once the buffers are inserted, to check whether a program can be translated, we can generate verification conditions and check if they are valid.

The verification condition generation algorithm typically requires annotations for loop invariants, but in our case, loop invariants are automatically inferred.
This is made possible by the restriction explained in Section~\ref{sec:auto-restriction}.
For each \(\mathit{low}_i\), \(\mathit{high}_i\) and \(s_i\) introduced in Section~\ref{sec:auto-restriction},
the tool prepares a linear template
\(c_0+c_1x_1+\cdots + c_\ell x_\ell\), where \(x_1,\ldots,x_\ell\) are integer variables occurring in the given
program\footnote{In the current implementation, we further restrict \(x_1,\ldots,x_\ell\) to loop variables,
but that restriction is easy to remove.}
and \(c_0,\ldots,c_\ell\) are unknown constants.
For example, for the second loop of Figure~\ref{fig:ex-stream}, the tool prepares the following template
for loop invariant, which we denote \( \Inv \):\footnote{For simplicity, here
\(N\) is treated as a constant not depending on \(b\).}
\begin{align*}
&\iset{\In} = [c_{0,0}+c_{0,1}x, c_{1,0}+c_{1,1}x; c_{2,0}+c_{2,1}x] \\
&\land \iset{\Out} = [c_{3,0}+c_{3,1}x, c_{4,0}+c_{4,1}x; c_{5,0}+c_{5,1}x] \land x \ne N-1
\end{align*}

We briefly explain how the verification conditions are generated.
As in the standard Hoare Logic, the verification condition is calculated using the approximate weakest precondition \( \awp s \BTE \), which is defined based on the translation rules in Section~\ref{sec:buffer}.
For example, the approximate weakest precondition for a write to an array is defined by
 \[
 \awp{a[e] := x} \BTE \defeq  e \notin \iset a  \land [a \set{e \mapsto x}/ a, \iset a \cdot e / \iset a] \BTE.
 \]
 Note that this is exactly the precondition of \rn{Tr-WriteMem}. In the end of the second loop of Figure~\ref{fig:ex-stream}, the (template) invariant \( [x + 1/x]\Inv \) must hold, and thus, the condition \( \awp{\mathit{out[x]:=z_1}}{[x + 1/x]\Inv} \equiv x \notin \iset \Out \land [\iset \Out \cdot x/ \iset \Out, x + 1/x]\Inv\) must hold before the statement \( out[x]:=z_1 \).
Thanks to the shape of the template, \( x \notin \iset \Out \) and \( \iset \Out \cdot x \) can be represented as formulas in the theory of (quantified) integer arithmetic.
For instance, \( x \notin \iset \Out \) is represented as
 \[
     \lnot \exists d (\low_\Out + d \times s_\Out = x \land \low_\Out \le x \le \high_\Out)
 \]
By chaining this backward reasoning, the actual verification condition for the loop is calculated as  \( Inv \land x < N \implies \awp {s} {[x + 1/x] \Inv} \), where \( s \) is the body of the loop.
This condition checks whether \( \Inv \) is a genuine loop invariant.\footnote{The actual verification condition also contains the entailment among pre-/post-condition and the invariant.}
For the example above,
\(c_{0,0}=c_{0,1}=c_{2,0}=c_{4,1}=c_{5,0}=1, c_{1,0}=N-1, c_{1,1}=c_{2,1}=c_{3,0}=c_{3,1}=c_{5,1}=0, c_{4,0}=-1\)
yields the loop invariant as shown on line 12 of Figure~\ref{fig:ex-stream} and makes the verification condition valid.
Then the tool transforms the program according to the translation rules and obtains the program in the right-hand-side of Figure~\ref{fig:ex-stream}.
If the verification condition is unsatisfiable, the tool gives up replacing (some of) the arrays with streams,
and keeps the arrays as they are.

\section{Experiment}
\label{sec:experiment}

\subsection{Implementation}

We have implemented a prototype tool based on the formalization of our translation presented in
Section~\ref{sec:buffer}.
As described in Section~\ref{sec:auto}, 
we have imposed certain restrictions to enable automated translation in the implementation.
To assess the individual contributions of buffer insertion and stream replacement to performance improvement, 
our implementation allows the stream replacement phase to be optionally skipped.
In what follows, we refer to the translation that performs only buffer insertion as \emph{buffer translation}, 
and the subsequent stream replacement step as \emph{stream translation}.
We have designed a C-like domain-specific language (DSL) as the source language of our tool, 
with \CC{} and hardware optimization directives as the target language.
The DSL supports (non-recursive) functions and two-dimensional arrays in addition to the primitives of the source language
presented in Section~\ref{sec:lang}.

\subsection{Setup}

To evaluate our proposal, we conduct computing performance measurements on a real FPGA platform using our prototype tool.
We use Z3~\cite{DBLP:conf/tacas/MouraB08} version 4.11.2 as the backend SMT solver to discharge verification conditions.
We use AMD Vitis HLS 2024.1 as the backend compiler of our tool to generate register transfer level (RTL) code from translated \CC{} source code with directives,
and we use AMD Vivado 2024.1 to generate an FPGA bitstream of the target hardware accelerator.
We use AMD Kria KV260 Vision AI Starter Kit~\cite{k26_kv260_amd}, an SoC FPGA consisting of a programmable logic (PL) part (i.e., reconfigurable logic) and processing system (PS) part (i.e., CPU cores) as the target FPGA platform.
Without the stream translation, off-chip memory accesses are conducted via the AXI Main (formerly Master) interface of the generated accelerator.
With the stream translation, the AXI Main interfaces are replaced with AXI Stream interfaces, if possible, and the AXI DMA controller transfers the input and output data to the accelerator in a burst manner.

We use PYNQ~\cite{kria_pynq_github}, a Python-based infrastructure providing software libraries to control hardware accelerators on the programmable logic from the integrated CPU.
The host software code, which consists of a non-kernel part of the translated program, including DMA requests, is executed on the CPU.

We use the following base benchmark programs for the performance measurements:

\begin{itemize}
    \item \textbf{Filter} : The motivating example shown in Figure~\ref{fig:intro1} in Section~\ref{sec:intro}.
    \item \textbf{Filter-Dilated} : A variant of \textbf{Filter} that computes the average of elements at distant indices.
    \item \textbf{Divide} : The motivating example shown in Figure~\ref{fig:intro-pitfall} in Section~\ref{sec:intro}.
    \item \textbf{MatVecMul} : A program that performs matrix-vector multiplication.
    \item \textbf{MatAdd} : A program that performs addition of two matrices.
    \item \textbf{Filter-2D} : A version of \textbf{Filter} that processes a two-dimensional array.
    \item \textbf{Stencil-2D} : A structured grid computation applying a fixed-size filter over a 2D array.
    \item \textbf{KMP} : A classic string-matching algorithm.
    \item \textbf{GeMM} : Multiplication of two dense matrices.
    \item \textbf{SpMV} : Multiplication of a sparse matrix and a dense vector.
\end{itemize}

\noindent
To evaluate the applicability of our framework to realistic workloads, 
we include several programs 
\footnote{\Revision{
We used Stencil-2D, KMP, GeMM, and SpMV from MachSuite. 
The remaining MachSuite benchmarks were not included because they fall outside the scope of our transformation. 
For example, programs such as AES perform only small, fixed memory lookups without reuse, 
so the translated program is identical to the original and provides no additional insight.
}}
from MachSuite~\cite{reagen2014machsuite}, a widely used benchmark suite for accelerator research.
\footnote{
\Revision{
We also attempted to evaluate existing source-to-source HLS frameworks,
such as Stream-HLS~\cite{10.1145/3706628.3708878}, as comparison targets.
However, the framework targets different input models (e.g., Torch-MLIR or
restricted affine programs) and could not be applied to our C/C++ benchmarks
in a fully automatic manner.
}}
Note that the MachSuite programs are originally written in \CC{}, so we manually simplified them to fit our DSL.
For each benchmark, we set the maximum input length to at least \(2^{18}\).

In addition, to show the capability of the proposal for non-trivial memory access patterns, we consider two variants for each base benchmark program, as described below.

\begin{itemize}
    \item \textbf{Rev} : A version that reads the input array in reverse order.
    \item \textbf{Skip} : A version in which the loop index is incremented by 2.
\end{itemize}
\Revision{
\noindent
These variants represent non-trivial memory access patterns that can arise in practical applications.
For instance, reverse traversal appears in backward scans (e.g., in dynamic programming or image processing), 
and skip-like patterns occur in signal and image processing tasks such as subsampling, decimation, or downsampling.
}

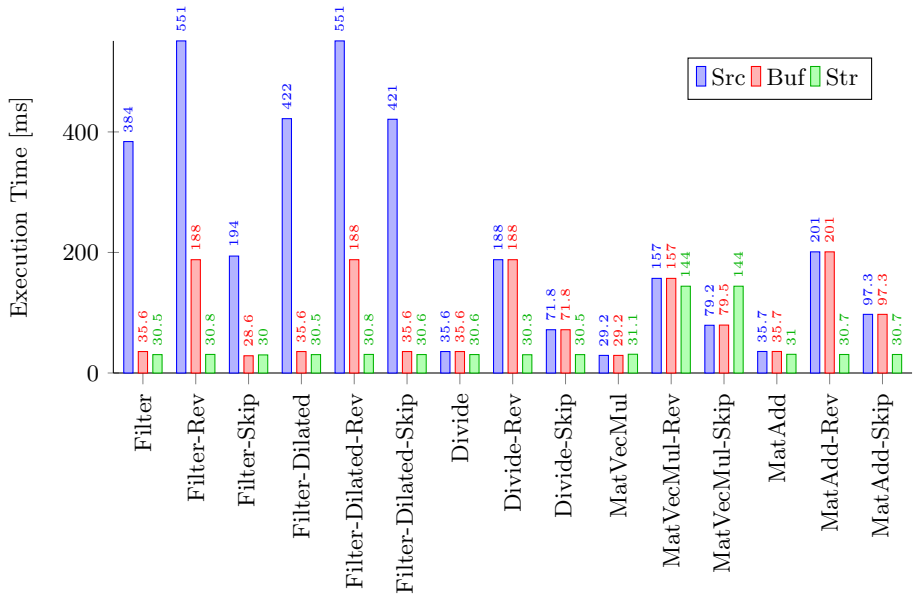
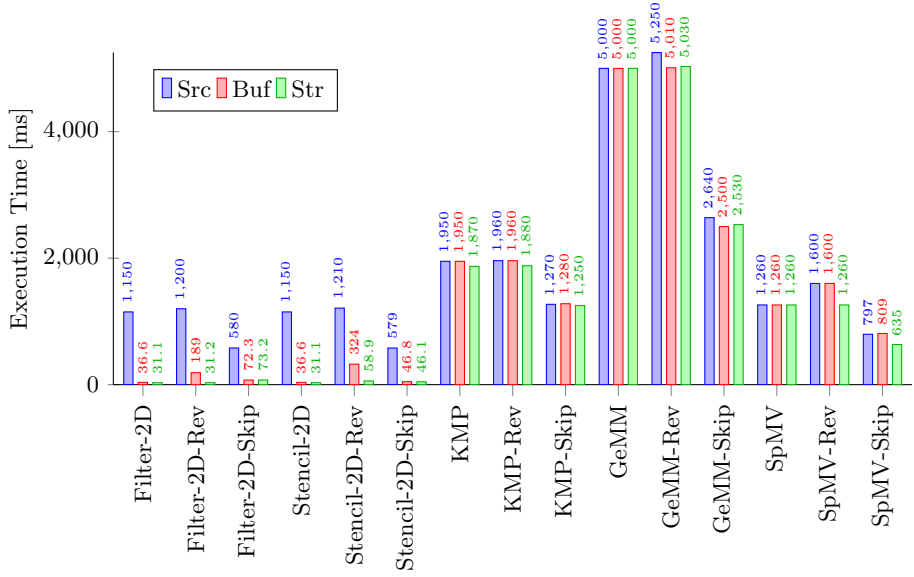
\begin{figure*}
\centering
\subfloat[]{
\begin{tikzpicture}
\begin{axis}[
    ybar,
    bar width=3.5pt,
    axis lines=left,
    axis line style={-},
    width=\linewidth,
    height=6cm,
    ylabel={Execution Time [ms]},
    symbolic x coords={
      Filter, Filter-Rev, Filter-Skip,
      Filter-Dilated, Filter-Dilated-Rev, Filter-Dilated-Skip,
      Divide, Divide-Rev, Divide-Skip,
      MatVecMul, MatVecMul-Rev, MatVecMul-Skip,
      MatAdd, MatAdd-Rev, MatAdd-Skip
    },
    xtick=data,
    x tick label style={rotate=90, anchor=east, font=\footnotesize},
    enlarge x limits=0.04,
    ymin=0,
    legend style={at={(0.95,0.95)}, anchor=north east, legend columns=-1},
    legend cell align={left},
    every node near coord/.append style={font=\tiny, rotate=90, anchor=west},
    nodes near coords,
    ybar,aan ybar legend
]
\addplot+[blue,fill=blue!30] coordinates {
  (Filter,384) (Filter-Rev,551) (Filter-Skip,194)
  (Filter-Dilated,422) (Filter-Dilated-Rev,551) (Filter-Dilated-Skip,421)
  (Divide,35.6) (Divide-Rev,188) (Divide-Skip,71.8)
  (MatVecMul,29.2) (MatVecMul-Rev,157) (MatVecMul-Skip,79.2)
  (MatAdd,35.7) (MatAdd-Rev,201) (MatAdd-Skip,97.3)
};
\addplot+[red,fill=red!30] coordinates {
  (Filter,35.6) (Filter-Rev,188) (Filter-Skip,28.6)
  (Filter-Dilated,35.6) (Filter-Dilated-Rev,188) (Filter-Dilated-Skip,35.6)
  (Divide,35.6) (Divide-Rev,188) (Divide-Skip,71.8)
  (MatVecMul,29.2) (MatVecMul-Rev,157) (MatVecMul-Skip,79.5)
  (MatAdd,35.7) (MatAdd-Rev,201) (MatAdd-Skip,97.3)
};
\addplot+[green!70!black,fill=green!30] coordinates {
  (Filter,30.5) (Filter-Rev,30.8) (Filter-Skip,30.0)
  (Filter-Dilated,30.5) (Filter-Dilated-Rev,30.8) (Filter-Dilated-Skip,30.6)
  (Divide,30.6) (Divide-Rev,30.3) (Divide-Skip,30.5)
  (MatVecMul,31.1) (MatVecMul-Rev,144) (MatVecMul-Skip,144)
  (MatAdd,31.0) (MatAdd-Rev,30.7) (MatAdd-Skip,30.7)
};
\legend{Src, Buf, Str}
\end{axis}
\end{tikzpicture}
}
\hfill
\subfloat[]{
\begin{tikzpicture}
\begin{axis}[
    ybar,
    bar width=3.5pt,
    axis lines=left,
    axis line style={-},
    width=\linewidth,
    height=6cm,
    ylabel={Execution Time [ms]},
    symbolic x coords={
      Filter-2D, Filter-2D-Rev, Filter-2D-Skip,
      Stencil-2D, Stencil-2D-Rev, Stencil-2D-Skip,
      KMP, KMP-Rev, KMP-Skip,
      GeMM, GeMM-Rev, GeMM-Skip,
      SpMV, SpMV-Rev, SpMV-Skip
    },
    xtick=data,
    x tick label style={rotate=90, anchor=east, font=\footnotesize},
    enlarge x limits=0.04,
    ymin=0,
    legend style={at={(0.05,0.95)}, anchor=north west, legend columns=-1},
    legend cell align={left},
    every node near coord/.append style={font=\tiny, rotate=90, anchor=west},
    nodes near coords,
    ybar,aan ybar legend
]
\addplot+[blue,fill=blue!30] coordinates {  
  (Filter-2D,1150) (Filter-2D-Rev,1200) (Filter-2D-Skip,580)
  (Stencil-2D,1150) (Stencil-2D-Rev,1210) (Stencil-2D-Skip,579)
  (KMP,1950) (KMP-Rev,1960) (KMP-Skip,1270)
  (GeMM,5000) (GeMM-Rev,5250) (GeMM-Skip,2640)
  (SpMV,1260) (SpMV-Rev,1600) (SpMV-Skip,797)
};
\addplot+[red,fill=red!30] coordinates {
  (Filter-2D,36.6) (Filter-2D-Rev,189) (Filter-2D-Skip,72.3)
  (Stencil-2D,36.6) (Stencil-2D-Rev,324) (Stencil-2D-Skip,46.8)
  (KMP,1950) (KMP-Rev,1960) (KMP-Skip,1280)
  (GeMM,5000) (GeMM-Rev,5010) (GeMM-Skip,2500)
  (SpMV,1260) (SpMV-Rev,1600) (SpMV-Skip,809)
};
\addplot+[green!70!black,fill=green!30] coordinates {
  (Filter-2D,31.1) (Filter-2D-Rev,31.2) (Filter-2D-Skip,73.2)
  (Stencil-2D,31.1) (Stencil-2D-Rev,58.9) (Stencil-2D-Skip,46.1)
  (KMP,1870) (KMP-Rev,1880) (KMP-Skip,1250)
  (GeMM,5000) (GeMM-Rev,5030) (GeMM-Skip,2530)
  (SpMV,1260) (SpMV-Rev,1260) (SpMV-Skip,635)
};
\legend{Src, Buf, Str}
\end{axis}
\end{tikzpicture}
}
\caption{Execution time of the kernel functions. 
    \emph{Src} indicates the execution time in the original source program,
    \emph{Buf} shows the time after applying buffer translation only, 
    and \emph{Str} shows the time after completing stream translation.
    Benchmarks are split across (a) and (b) for layout reasons.
}
\label{fig:experiment-bar}
\end{figure*}


\subsection{Evaluation and Discussion}

Figure~\ref{fig:experiment-bar} shows the execution times of the kernel functions of the benchmark programs.
Note that translation time of each program by our tool was less than a second.

First, we begin the discussion for the base benchmarks.
The programs \textbf{Filter}, \textbf{Filter-Dilated}, \textbf{Filter-2D}, 
and \textbf{Stencil-2D} exhibited significant performance improvements thanks to the buffer translation.
The performance improvements can be attributed to two main factors:
(1) a reduction in the number of costly off-chip memory accesses, and
(2) the simplification of memory access patterns, which enabled Vitis HLS to infer burst transfers and obtain the required data in bulk from the off-chip memory.

A burst transfer delivers multiple words of contiguous addresses on the off-chip memory for each transaction, unlike a single transfer, which delivers only one word.
Utilizing burst transfers is crucial for efficient memory bandwidth utilization and effective computing performance.
The backend high-level synthesis compiler (Vitis HLS in this experiment) can infer burst transfers through identifications of sequential memory accesses in a program, and multiple words are transferred at high throughput between on-chip fabrics and the off-chip memory.
However, original access patterns to the off-chip memory in \textbf{Filter}, \textbf{Filter-Dilated}, \textbf{Filter-2D},
and \textbf{Stencil-2D} are not in sequential order, and no burst transfer can be synthesized.

In our approach, the buffer translation inserts on-chip reuse buffers using block RAMs in the FPGA for marshaling access orders to the off-chip memory.
The first memory access to a word is still represented as off-chip memory access, but the obtained word is stored on the inserted on-chip reuse buffer.
Then, subsequent memory accesses are replaced with on-chip reuse buffer accesses, and the corresponding off-chip memory accesses are eliminated.
Finally, the remaining off-chip memory accesses become in sequential order, so the backend high-level synthesis compiler can successfully synthesize burst transfers.

In contrast, the performance of other base benchmarks remained largely unchanged.
For the programs \textbf{Divide}, \textbf{MatVecMul}, \textbf{MatAdd}, \textbf{KMP}, and \textbf{SpMV},
the original memory access patterns are already in sequential order.
Consequently, the buffer translation does not insert any on-chip reuse buffers for these programs.
As a result, the backend compiler can synthesize efficient burst transfers in the original programs, and no additional performance benefit was observed.
For \textbf{GeMM}, although line buffers are introduced, at least one input matrix is accessed in a non-linear order.
Therefore, the buffer translation does not insert reuse buffers for that matrix 
and the access to it likely constitutes the performance bottleneck.

Next, we consider the variant programs.
The performance of the \textbf{Rev} variants, except for \textbf{MatVecMul}, \textbf{KMP} and \textbf{GeMV}, was improved with the stream translation.
In their original forms, these programs exhibited sequential but reverse memory access patterns, which prevent the backend compiler from inferring burst accesses.
In contrast, the stream translation removed the indexing parts of off-chip memory accesses from the accelerator, and off-chip memory accesses were conducted via the AXI DMA controller instead.
Note that the off-chip data transfer behavior and performance of the AXI DMA controller are the same as those of the burst transfer by the AXI Main interface, and the maximum off-chip bandwidth utilization by the AXI DMA controller is equivalent to that of the burst transfer.
Here, the stream translation changes the access order to the normal sequential order for the AXI DMA controller because the DMA controller does not support the reverse order.
Finally, the translated programs after the stream translation can receive and send the input and output data to/from the off-chip memory in the burst mode at high bandwidth utilization.

Unfortunately, in the case of \textbf{MatVecMul-Rev}, the performance remained unchanged.
The vectors in this program are accessed in non-linear order, making them unsuitable for the stream translation. 
The existence of such non-linear accesses prevents our method from converting the program into the stream style and the backend compiler from inferring efficient burst transfers. 
For \textbf{KMP-Rev} and \textbf{GeMM-Rev}, although the translation enables the backend compiler to infer burst transfers,
other parts of the programs likely became performance bottlenecks, and thus no improvements were observed.

In the cases of \textbf{Skip} variants, the stream translation provides performance improvements for \textbf{Divide-Skip}, \textbf{MatAdd-Skip}, and \textbf{SpMV-Skip}.
After the stream translation, in addition to the essential words in the computation, unused words in the original program are transferred from/to the off-chip memory so that the AXI DMA controller can utilize the burst transfer.
Finally, the data transfer throughput and computing performance were successfully improved while the number of transferred words was increased.

In contrast, the stream translation does not provide sufficient performance improvements for the \textbf{Skip} variants of \textbf{Filter}, \textbf{Filter-Dilated}, \textbf{Filter-2D} and \textbf{Stencil-2D} compared to their \textbf{Rev} counterparts. 
This is because the buffer translation before the stream translation already marshaled the off-chip memory accesses very well, and the off-chip memory access patterns are contiguous, allowing the backend compiler to infer burst transfers.
As a result, there are no additional benefits to the stream translation.

\begin{table*}[t]
  \caption{
    Area and power consumption of the kernel functions, 
    as reported by Vivado after synthesizing the corresponding hardware implementations.
    \Revision{\emph{Src}, \emph{Buf}, and \emph{Str} denote the original source kernel,
    the buffered version, and the stream-replaced version, respectively.}
  }
  \label{table:experiments2}
  \centering
  \resizebox{\textwidth}{!}{
  \begin{tabular}{l|ccc|ccc|ccc|ccc|ccc|ccc}
    \hline
    Name & \multicolumn{3}{c|}{Power [W]} & \multicolumn{3}{c|}{LUT} & \multicolumn{3}{c|}{LUTRAM} & \multicolumn{3}{c|}{FF} & \multicolumn{3}{c|}{BRAM} & \multicolumn{3}{c}{DSP} \\
         & \emph{Src} & \emph{Buf} & \emph{Str} & \emph{Src} & \emph{Buf} & \emph{Str} & \emph{Src} & \emph{Buf} & \emph{Str} & \emph{Src} & \emph{Buf} & \emph{Str} & \emph{Src} & \emph{Buf} & \emph{Str} & \emph{Src} & \emph{Buf} & \emph{Str} \\
    \hline\hline
\textbf{Filter}  & 2.748 & 2.751 & 2.746 & 3003 & 3161 & 4003 & 473 & 485 & 515 & 3575 & 3764 & 5007 & 1 & 1 & 5 & 4 & 4 & 4 \\
\textbf{Filter-Rev} & 2.746 & 2.752 & 2.746 & 3000 & 3139 & 4003 & 467 & 481 & 515 & 3533 & 3697 & 5007 & 1 & 1 & 5 & 4 & 4 & 4 \\
\textbf{Filter-Skip} & 2.748 & 2.750 & 2.746 & 2980 & 3041 & 3951 & 462 & 486 & 515 & 3566 & 3730 & 4940 & 1 & 1 & 5 & 4 & 4 & 4 \\
\hline
\textbf{Filter-Dilated} & 2.748 & 2.752 & 2.746 & 3070 & 3218 & 4012 & 471 & 483 & 515 & 3677 & 3883 & 5137 & 1 & 1 & 5 & 4 & 4 & 4 \\
\textbf{Filter-Dilated-Rev} & 2.749 & 2.752 & 2.746 & 3029 & 3138 & 4012 & 466 & 479 & 515 & 3616 & 3810 & 5137 & 1 & 1 & 5 & 4 & 4 & 4 \\
\textbf{Filter-Dilated-Skip} & 2.748 & 2.752 & 2.746 & 3052 & 3131 & 3984 & 460 & 472 & 515 & 3667 & 3702 & 4974 & 1 & 1 & 5 & 4 & 4 & 4 \\
\hline
\textbf{Divide} & 2.747 & 2.746 & 2.743 & 2792 & 2799 & 3819 & 478 & 480 & 515 & 3416 & 3420 & 4839 & 1 & 1 & 5 & 0 & 0 & 0 \\
\textbf{Divide-Rev} & 2.748 & 2.748 & 2.743 & 2834 & 2843 & 3819 & 476 & 478 & 515 & 3444 & 3448 & 4839 & 1 & 1 & 5 & 0 & 0 & 0 \\
\textbf{Divide-skip} & 2.747 & 2.747 & 2.743 & 2803 & 2815 & 3812 & 461 & 463 & 515 & 3414 & 3418 & 4839 & 1 & 1 & 5 & 0 & 0 & 0 \\
\hline
\textbf{MatVecMul} & 2.758 & 2.758 & 2.756 & 3560 & 3567 & 4546 & 655 & 657 & 677 & 4395 & 4399 & 5837 & 1.5 & 1.5 & 5.5 & 3 & 3 & 3 \\
\textbf{MatVecMul-Rev} & 2.760 & 2.759 & 2.757 & 3636 & 3685 & 4574 & 658 & 664 & 705 & 4452 & 4464 & 5874 & 1.5 & 1.5 & 5.5 & 3 & 3 & 3 \\
\textbf{MatVecMul-Skip} & 2.759 & 2.759 & 2.756 & 3610 & 3625 & 4580 & 651 & 657 & 721 & 4408 & 4420 & 5939 & 1.5 & 1.5 & 5.5 & 3 & 3 & 3 \\
\hline
\textbf{MatAdd} & 2.759 & 2.759 & 2.753 & 3487 & 3496 & 4777 & 650 & 654 & 701 & 4404 & 4412 & 6155 & 1.5 & 1.5 & 7.5 & 0 & 0 & 0 \\
\textbf{MatAdd-Rev} & 2.761 & 2.759 & 2.753 & 3598 & 3646 & 4777 & 648 & 652 & 701 & 4469 & 4477 & 6155 & 1.5 & 1.5 & 7.5 & 0 & 0 & 0 \\
\textbf{MatAdd-Skip} & 2.759 & 2.760 & 2.753 & 3571 & 3584 & 4736 & 650 & 654 & 701 & 4439 & 4447 & 6152 & 1.5 & 1.5 & 7.5 & 0 & 0 & 0 \\
\hline
\textbf{Filter-2D} & 2.752 & 2.818 & 2.787 & 3205 & 5835 & 6869 & 484 & 623 & 733 & 3893 & 5911 & 7182 & 1 & 53 & 57 & 4 & 4 & 4 \\
\textbf{Filter-2D-Rev} & 2.754 & 2.818 & 2.787 & 3377 & 5985 & 6896 & 459 & 659 & 733 & 3922 & 5987 & 7182 & 1 & 53 & 57 & 4 & 4 & 4 \\
\textbf{Filter-2D-Skip} & 2.753 & 2.790 & 2.779 & 3228 & 5752 & 6620 & 467 & 674 & 745 & 3962 & 5970 & 7189 & 1 & 7.5 & 11.5 & 4 & 4 & 4 \\
\hline
\textbf{Stencil-2D} & 2.764 & 2.864 & 2.814 & 4022 & 6545 & 7598 & 645 & 800 & 896 & 5143 & 6868 & 8266 & 1.5 & 65.5 & 69.5 & 3 & 27 & 27 \\
\textbf{Stencil-2D-Rev} & 2.762 & 2.854 & 2.806 & 4119 & 6790 & 7748 & 635 & 844 & 916 & 5200 & 6841 & 8210 & 1.5 & 65.5 & 69.5 & 3 & 27 & 27 \\
\textbf{Stencil-2D-Skip} & 2.760 & 2.853 & 2.819 & 3969 & 6664 & 7711 & 645 & 807 & 903 & 5138 & 6721 & 8141 & 1.5 & 65.5 & 69.5 & 3 & 27 & 27 \\
\hline
\textbf{Kmp} & 2.759 & 2.759 & 2.755 & 3621 & 3632 & 4139 & 601 & 607 & 626 & 4514 & 4526 & 5255 & 1.5 & 1.5 & 3.5 & 0 & 0 & 0 \\
\textbf{Kmp-Rev} & 2.760 & 2.756 & 2.755 & 3750 & 3737 & 4218 & 601 & 607 & 626 & 4524 & 4559 & 5297 & 1.5 & 1.5 & 3.5 & 0 & 0 & 0 \\
\textbf{Kmp-Skip} & 2.756 & 2.757 & 2.754 & 3712 & 3733 & 4216 & 601 & 607 & 626 & 4544 & 4556 & 5295 & 1.5 & 1.5 & 3.5 & 0 & 0 & 0 \\
\hline
\textbf{GeMM} & 2.759 & 2.762 & 2.758 & 3681 & 3630 & 4628 & 655 & 646 & 713 & 4447 & 4526 & 5967 & 1.5 & 2 & 6 & 3 & 3 & 3 \\
\textbf{GeMM-Rev} & 2.761 & 2.762 & 2.760 & 3693 & 3728 & 4663 & 671 & 674 & 716 & 4497 & 4647 & 5973 & 1.5 & 2 & 6 & 3 & 3 & 3 \\
\textbf{GeMM-skip} & 2.759 & 2.761 & 2.758 & 3676 & 3616 & 4643 & 655 & 655 & 716 & 4443 & 4553 & 5980 & 1.5 & 2 & 6 & 3 & 3 & 3 \\
\hline
\textbf{SpMV} & 2.782 & 2.782 & 2.776 & 5147 & 5148 & 5873 & 979 & 983 & 1023 & 6531 & 6539 & 7545 & 2.5 & 2.5 & 6.5 & 3 & 3 & 3\\
\textbf{SpMV-Rev} & 2.783 & 2.783 & 2.777 & 5233 & 5244 & 5876 & 979 & 983 & 1027 & 6531 & 6539 & 7529 & 2.5 & 2.5 & 6.5 & 3 & 3 & 3 \\
\textbf{SpMV-Skip} & 2.783 & 2.783 & 2.777 & 5247 & 5275 & 5911 & 979 & 983 & 1027 & 6530 & 6583 & 7599 & 2.5 & 2.5 & 6.5 & 3 & 3 & 3 \\

    \hline
  \end{tabular}
  }
\end{table*}

Table~\ref{table:experiments2} reports the area and power consumption of the generated hardware implementations.
\Revision{Here, LUT, LUTRAM, FF, BRAM, and DSP denote standard FPGA resources:
look-up tables, LUT-based RAMs, flip-flops, block RAMs, and DSP blocks, respectively.}
It indicates that the use of additional buffers moderately increases area but has little effect on power consumption.
For all benchmark instances, except for \textbf{Filter-2D} and \textbf{Stencil-2D}, the buffer translation does not significantly change the area, and the stream translation adds only a small overhead in BRAM usage, presumably reflecting the cost of FIFO buffers used for stream interfaces.
In \textbf{Filter-2D} and \textbf{Stencil-2D}, the buffer translation leads to a noticeable increase in BRAM usage, likely due to the insertion of line buffers for handling two-dimensional data.

To summarize our evaluation, the buffer translation can change the off-chip memory access patterns to sequential order by using on-chip reuse buffers, and data transfer throughputs are improved thanks to the effective use of burst transfers.
In addition, the stream translation is effective for non-trivial memory access patterns by enforcing burst transfers using the DMA controller and stream processing even with redundant data transfers and computations.

\section{Related Work}
\label{sec:related}

\subsubsection*{Program Transformation for High-Level Synthesis}

Several program transformation approaches for high-level synthesis (HLS) have been proposed, although none are based on relational Hoare logics.
\Revision{
Seto et al.~\cite{seto2018scalar,seto2019scalar} optimize affine C programs for HLS using scalar replacement and the polyhedral model, 
improving performance and reducing hardware area.  
POLSCA~\cite{zhao2022polsca}, ScaleHLS~\cite{9773203}, HIDA~\cite{ye2024hida}, POM~\cite{zhang2024optimizing}, 
and Stream-HLS~\cite{10.1145/3706628.3708878} are MLIR-based transformation frameworks that target affine programs and apply
polyhedral analysis, multi-level design abstractions, or dataflow-oriented restructuring to generate optimized, HLS-friendly code.
These approaches do not provide formal guarantees of semantic preservation for the applied program transformations.}
In contrast, our method differs significantly:
by leveraging relational Hoare logic, it supports more flexible transformations
and enables reasoning about precise invariants required for tricky transformations,
providing formal guarantees of correctness.
Although we considered only array types in this paper, the use of types will allow us to handle non-array data structures and
control structures such as recursion.
In fact, in a different context, Kodama et al.~\cite{kodama2004translation,sato2011ordered}
apply ordered linear types to transform tree-processing programs to stream-processing programs.
We expect that a similar idea can be applied to compile tree-processing HLS programs to efficient hardware
accelerators.

In Section~\ref{sec:auto}, we have restricted array access information (represented by \(\iset{a}\))
to certain shapes; this has been partially inspired by previous work on array access analysis
in the context of compiler optimizations~\cite{paek2002efficient}.

\subsubsection*{DSE-Based Approaches to High-Level Synthesis}

\Revision{
Several design space exploration (DSE)-based approaches~\cite{zhao2017comba,pouget2025automatic,yu2018s2fa,huang2021pylog}}
transform high-level programs into intermediate representations and perform automated exploration of
the design space, such as loop pipelining and array partitioning parameters.
These methods assume hardware-friendly input and focus on tuning synthesis parameters for optimal performance.
In contrast, our method automatically transforms naïve high-level programs into hardware-friendly forms.
It could potentially serve as a preprocessing step for DSE frameworks, broadening their applicability to less optimized input.

\subsubsection*{Formal Approaches to High-Level Synthesis}


Nigam et al.~\cite{nigam2020predictable} propose a type system based on affine types that
restricts programs to hardware-friendly forms, effectively limiting the design space to efficient patterns.
Pouchet et al.~\cite{10.1145/3626202.3637563} develop a program equivalence checker for HLS optimizations, based on abstract interpretation and symbolic analysis.
While these studies assist in writing efficient HLS code, they are not optimizers for naïve programs.
A code that type-checks or an optimized code that is given to the equivalence checker needs to be written by humans or generated by external tools.

\subsubsection*{Relational Hoare Logics}

Relational Hoare logics have been drawing attention recently 
in the context of relational program verification~\cite{DBLP:journals/pacmpl/AvanziniBDG25,DBLP:conf/lics/NagasamudramN21,barthe2009formal}:
see~\cite{DBLP:conf/isola/Naumann20} for a survey.
They have also been used to establish the correctness of program transformations~\cite{10.5555/889016},  
but to our knowledge, have not been applied in the context of high-level synthesis for hardware acceleration.  
The specific design of relational Hoare logic, including index-set variables and other features, is also novel and tailored for automatic translation to HLS programs.

\section{Conclusion}
\label{sec:concl}

We have proposed a formally defined framework for translating naïve programs into efficient ones for high-level synthesis (HLS), 
using relational Hoare logic to ensure robustness to programming style. 
We have implemented a prototype tool that automatically translates C-like programs to HLS \CC{},
and confirmed that the translation improved the hardware performance of certain programs, achieving
up to 38.5x speed-ups.

\subsubsection*{Future Work}
Extending our method to support more complex data such as trees and control structures is an important next step.
We also plan to allow user annotations of loop invariants, to further leverage the expressive power of relational Hoare logic
and enable more flexible transformations.
Beyond these extensions, our RHL-based framework has potential applications beyond buffer insertion and stream processing.
It could help reason about hardware-level constraints---for example, suggesting tiling when the inferred buffer size 
exceeds on-chip memory capacity---or be applied to scheduling to fully utilize FPU pipelines and coordinate 
multiple functional units while preserving data dependencies.
These directions demonstrate the broader applicability of our RHL-based approach.

\subsubsection*{Acknowledgments}
    We would like to thank anonymous reviewers for useful comments.
    This work was supported by JSPS KAKENHI Grant Number JP20H05703, 
    and JST SPRING Grant Number JPMJSP2108.

\bibliographystyle{splncs04}
\bibliography{myref}

\iffull
\clearpage
\appendix
\let\oldinfer\infer
\renewcommand*{\infer}[2]{\oldinfer{#2}{#1}}
\section{Correctness of the Translation}
\label{appx:correctness-stream}
\colorlet{assertioncolor}{magenta}
\newcommand*{\fml}{\textcolor{assertioncolor}{\BTE}}
\newcommand*{\sub}[2]{[#1 / #2]}
\newcommand*{\arTm}{\textcolor{assertioncolor}{\mathtt{a}}}
\newcommand*{\inTm}{\textcolor{assertioncolor}{\mathtt{e}}}
\newcommand*{\idxTm}{\textcolor{assertioncolor}{\mathcal{I}}}
\newcommand*{\sem}[1]{\llbracket #1 \rrbracket}
\newcommand*{\defiff}{\stackrel{\mathrm{def}}{\iff}}
\newcommand*{\relRHS}[7]{\sstate{#1}{#2} \sim_{#5\mid#6}^{#7} \sstate{#3}{#4}}
\newcommand*{\nil}{\mathsf{nil}}
\renewcommand*{\epsilon}{\varepsilon}
\subsection{Assertion language}
\label{appx:correctness-stream:assertion-lang}
We define the syntax and the semantics of the assertion language that was only informally given in the body of the paper.

The grammar of the assertion language as follows:
\begin{align*}
    \text{(array term)} \quad \arTm &\Coloneqq a \mid \updateAr \arTm {\inTm_1} {\inTm_2} \\
  \text{(integer term)} \quad \inTm &\Coloneqq e \mid \arTm[\inTm] \mid \inTm_1 \mathbin{\OP}  \inTm_2  \mid \hd(\idxTm) \\
  \text{(index sequence term)} \quad \idxTm &\Coloneqq \iset a  \mid \nil \mid \inTm \cdot \idxTm  \mid \idxTm \cdot \inTm  \mid \tl(\idxTm) \\
  \text{(formula)} \quad \fml &\Coloneqq  \TRUE \mid \inTm_1 = \inTm_2 \mid \inTm_1 \le \inTm_2 \mid \inTm \in \idxTm  \\
  &\mid \fml_1 \land \fml_2 \mid \neg \fml
\end{align*}
Here \( a \) and \( e \) are array variables and expression that are used in the syntax of source or target language, respectively.
The term \( \iset a \) represents the \emph{sequence} of indices of an array variable \( a \); intuitively it represents the access pattern of the array \( a \).
The constant \( \nil \) represents the empty list; and \( \inTm \cdot \idxTm  \) and \( \idxTm \cdot \inTm \) respectively represent concatenating an integer \( \inTm \) to the head and the tail of the list \( \idxTm \).
As expected, the operations \( \hd \) and \( \tl \) return the head and tail of a sequence.
The intuitive meaning for the rest of the constructors should be clear; see below for their formal interpretation.

\sk{It's strange to mix typewrite fonts and calligraphic fonts. I'll fix this later}

We now define the semantics.
We first give the interpretation of terms.
As expected, array terms are interpreted as a partial map from integers to integers, integer terms are interpreted as integers and index set terms are interpreted as sets of integers.
The semantics depends on three parameters \( R \), \( H \) and \( \Iseq \).
Recall that \( \Iseq \) is a finite map from array variables to integer sequences whose elements are pairwise distinct.
It is used to give the meaning of \( \iset a \).

\begin{align*}
  \sem{a}_{R, H, \Iseq}(m) &\defeq H(a, m) \\
  \sem{\updateAr \arTm {\inTm_1} {\inTm_2}}_{R, H, \Iseq}(m) &\defeq
    \begin{cases}
      \sem{\inTm_2}_{R, H, \Iseq} & \text{if \( \sem{\inTm_1}_{R, H, \Iseq} = m \)} \\
      \sem{\arTm}_{R, H, \Iseq}(m) & \text{otherwise}
    \end{cases} \\
\end{align*}
\begin{align*}
  \sem{e}_{R, H, \Iseq} &\defeq
    \begin{cases}
      n & \text{if \( \econfig R e \eval n \)} \\
      \text{undefined} &\text{otherwise}
    \end{cases}\\
  \sem{\arTm[\inTm]}_{R, H, \Iseq} &\defeq \sem{\arTm}_{R, H, \Iseq}(\sem{\inTm}_{R, H, \Iseq}) \\
  \sem{\inTm_1 \mathbin{\OP} \inTm_2}_{R, H, \Iseq} &\defeq \sem{\OP}(\sem{\inTm_1}_{R, H, \Iseq}, \sem{\inTm_2}_{R, H,\Iseq}) \\
  \sem{\hd(\idxTm)}_{R, H, \Iseq} &\defeq
  \begin{cases}
    \text{undefined} \quad \text{if \( \sem{\idxTm}_{R, H, \Iseq} = \varepsilon\)} \\
    m \quad \text{if \( \sem{\idxTm}_{R, H, \Iseq} = m \cdot \seq n \)}
  \end{cases}
\end{align*}

\begin{align*}
  \sem{\iset a}_{R, H, \Iseq} &\defeq \Iseq(a) \\
  \sem{\nil}_{R, H, \Iseq} &\defeq \epsilon \\
  \sem{\inTm \cdot \idxTm}_{R, H, \Iseq} &\defeq \sem{\inTm}_{R, H, \Iseq} \cdot \sem{\idxTm}_{R, H, \Iseq} \\
 \sem{\idxTm \cdot \inTm}_{R, H, \Iseq}  &\defeq \sem{\idxTm}_{R, H, \Iseq} \cdot \sem{\inTm}_{R, H, \Iseq} \\
 \sem{\tl(\inTm)}_{R, H, \Iseq} &\defeq
 \begin{cases}
   \text{undefined} \quad \text{if \( \sem{\idxTm}_{R, H, \Iseq} = \epsilon \)} \\
   \seq n \quad \text{if \( \sem{\idxTm}_{R, H, \Iseq} = m \cdot \seq n\)}
 \end{cases}
\end{align*}
Note that the interpretation \( \sem{\arTm}_{R, H, \Iseq} \) and \( \sem{\inTm}_{R, H, \Iseq} \) do not rely on \( \Iseq \) unless they contain \( \hd(\inTm)\) as a subterm.

The satisfaction relation \( R, H, \Iseq \models \fml \) is defined in a standard manner
\begin{align*}
  R, H, \Iseq \models \inTm \in \idxTm \defiff \sem{\inTm}_{R, H, \Iseq} \in \setof{\sem{\idxTm}_{R, H, \Iseq}}
\end{align*}
\begin{align*}
  R, H, \Iseq \models \TRUE &\text{ always holds} \\
  R, H, \Iseq \models \inTm_1 = \inTm_2 &\defiff \sem{\inTm_1}_{R, H, \Iseq} = \sem{\inTm_2}_{R, H, \Iseq} \\
  R, H, \Iseq \models \inTm_1 \le \inTm_2 &\defiff \sem{\inTm_1}_{R, H, \Iseq} \le \sem{\inTm_2}_{R, H, \Iseq} \\
  R, H, \Iseq \models \inTm \in \idxTm &\defiff \sem{\inTm}_{R, H, \Iseq} \in \setof{\sem{\idxTm}_{R, H, \Iseq}} \\
  R, H, \Iseq \models \fml_1 \land \fml_2 &\defiff R, H, \Iseq \models \fml_1 \text{ and } R, H, \Iseq \models \fml_2 \\
  R, H, \Iseq \models \neg \fml  &\defiff R, H, \Iseq \not \models \fml
\end{align*}
Here \( \mathit{setof} \) is simply a cast operator that converts a sequence of indices to a set of sequences.
In the sequel we may simply use the notation \( n \in \seq{m} \), to make the notation concise.

\begin{remark}[Partiality of the interpretation]
  \label{rmk:partiality}
  Since \( \sem{a}_{R, H, \Iseq} \) is defined as a partial map and we use partial operations such as \( \hd \) and \( \tl\) the interpretation of terms may be undefined.
  We stipulate that if a formula contains a term whose interpretation is undefined, then its semantics is false.
  Another way to deal with this partiality is to define the interpretation of array variables as a total function by using a dummy value, say \( 0 \), for the case \( a[e] \) is undefined, and add a formula for boundary check such as \( e \in \iset a \) for each array access; similarly we can do a ``nil check'' before taking the head or the tail of a sequence.
  While such encoding might be needed when we invoke SMT solvers, here, we do not take this approach to keep the presentation simple.
\qed
\end{remark}

We list some substitution lemmas that will be used in the sequel.

\begin{lemma}
\label{lem:assertion-subst-strm}
  Let \( \inTm \) be an integer term and \( x \) be an integer variable.
  Suppose that \( \sem{\inTm}_{R, H, \Iseq} = n \).
  Then we have
  \begin{enumerate}
    \item \( \sem{\sub \inTm  x \arTm}_{R, H, \Iseq} =  \sem{\arTm}_{R \set{x \mapsto n}, H, \Iseq}\)
    \item \( \sem{\sub \inTm x \inTm_2}_{R, H, \Iseq} =  \sem{\inTm_2}_{R \set{x \mapsto n}, H, \Iseq}\),
    \item \( \sem{\sub \inTm x \idxTm}_{R, H, \Iseq} =  \sem{\idxTm}_{R \set{x \mapsto n}, H, \Iseq}\), and
    \item \( R, H, \Iseq \models \sub \inTm x \fml \) if and only if \( R \set{x \mapsto n}, H, \Iseq \models \fml \).
  \end{enumerate}
  \sk{Provided that the interpretations are well-defined.}
\end{lemma}
\begin{proof}
  By induction on the structure of terms and formulas.
  The case (1), (2) (3) are proved simultaneously, and the proof of (4) uses (1) to (3).
  \qed
\end{proof}
\begin{lemma}
\label{lem:assertion-subst-arr-strm}
  Let \( \inTm_1 \) and \( \inTm_2 \) be an integer term and \( a \) be an array variable.
  Suppose that \( \sem{\inTm_1}_{R, H, \Iseq} = m \) and \( \sem{\inTm_1}_{R, H, \Iseq} = n \).
  Then we have
  \begin{enumerate}
    \item \( \sem{\sub {\updateAr a {\inTm_1} {\inTm_2}}  a \arTm}_{R, H, \Iseq} =  \sem{\arTm}_{R, H\set{(a, m) \mapsto n}, \Iseq}\)
    \item \( \sem{\sub {\updateAr a {\inTm_1} {\inTm_2}} a \inTm_3}_{R, H, \Iseq} =  \sem{\inTm_3}_{R, H\set{(a, m) \mapsto n}, \Iseq}\),
    \item \( \sem{\sub {\updateAr a {\inTm_1} {\inTm_2}} a \idxTm}_{R, H, \Iseq} =  \sem{\idxTm}_{R, H\set{(a, m) \mapsto n}, \Iseq}\), and
    \item \( R, H, \Iseq \models \sub {\updateAr a {\inTm_1} {\inTm_2}} a \fml \) if and only if \(R, H\set{(a, m) \mapsto n}, \Iseq \models \fml \).
  \end{enumerate}
  \sk{Provided that the interpretations are well-defined.}
\end{lemma}
\begin{proof}
  Similar to the previous lemma.
  \qed
\end{proof}
\begin{lemma}
  \label{lem:assertion-subst-index-strm}
  \noindent
  \begin{enumerate}
    \item Suppose that \( \Iseq(a) \) is nonempty. Then we have
      \begin{enumerate}
      \item \( \sem{\sub {(\tl(\iset a)} {\iset a} \idxTm}_{R, H, \Iseq} =  \sem{\idxTm}_{R, H, \Iseq \set{ a \mapsto \tl(\Iseq(a))}}\) and
      \item \( R, H, \Iseq \models \sub {\tl(\iset a)} {\iset a} \fml \) if and only if \(R, H, \Iseq \set{ a \mapsto \tl(\Iseq(a))} \models \fml \).
      \end{enumerate}
    \item Suppose the \( \sem{\inTm}_{R, H, \Iseq} = n\). Then we have
      \begin{enumerate}
      \item \( \sem{\sub {\iset a \cdot \inTm} {\iset a} \idxTm}_{R, H, \Iseq} =  \sem{\idxTm}_{R, H, \Iseq \set{a \mapsto \Iseq(a) \cdot n}}\) and
      \item \( R, H, \Iseq \models \sub {\iset a \cdot  \inTm} {\iset a} \fml \) if and only if \(R, H, \Iseq \set{a \mapsto \Iseq(a) \cdot n}\models \fml \).
      \end{enumerate}
  \end{enumerate}
\end{lemma}
\begin{proof}
  By induction on the structure of the index sequence terms and formulas.
  \qed
\end{proof}

\subsection{Proof of the Correctness}
\paragraph{Notation.} For convenience, the simulation relation \( \relRH R H {R'} {S'} \ATE \fml \), will sometimes be written as \( \relRHS R H {R'} {S'} \ATE \fml \Iseq \) exhibiting the \( \Iseq \), which is a witness of the simulation relation.
To avoid the overuse of \( = \), we write \( \equiv \) for equality of syntactical terms such as integer terms and formulas.
For example, we will write \( \fml \equiv b = a[0] \); here the \( \equiv \) represents the syntactical equality whereas \( = \) is the equality symbol of the first-order logic.

We start with proving some auxiliary lemmas.

\begin{lemma}
  \label{lem:relRH-preserves-eval-expr}
  Suppose that \( \relRH{R}{H}{R'}{S}{\ATE}{\fml} \) and \( \ATE \vdash e : \INT \).
  Then we have \( \econfig R e \eval n \) if and only if \( \econfig {R'} e \eval n \).
\end{lemma}
\begin{proof}
  By a straightforward induction on the derivation of the evaluations.
  \qed
\end{proof}

The following is a correctness result for the auxiliary relation \( \bjs{\ATE}{\fml}{\fml'}{\iexp} \).
While this is stated as a lemma it may be considered as one of the key parts of the proof of the correctness.
\begin{lemma}
  \label{lem:buf-trans-insert-respects-invariants}
  Assume that \( \bjs{\ATE}{\fml}{\fml'}{\iexp} \) and \( \relRH{R}{H}{R_1}{S_1}{\ATE}{\fml} \).
  Then there uniquely exist \( R_2\) and \( S_2\) such that \( \sconfig {R_1} {S_1} u \eval \sstate {R_2}{S_2}\), and moreover, \( \relRH{R}{H}{R_2}{S_2}{\ATE}{\fml'} \).
\end{lemma}
\begin{proof}
  By induction on the derivation of \( \bjs{\ATE}{\fml}{\fml'}{\iexp} \) with a case analysis on the last rule used.
  \begin{description}
    \item[\rn{Tr-InsRBuf}]
      It must be the case that \( \iexp \equiv b := \RStr a \), where \( \ATE(a) = \RARR \) and \( \ATE(b) = \BUF\), and
      \begin{align*}
        \fml = [a[\hd(\iset{a})]/b,(\tl(\iset{a}))/\iset{a}]\fml'.
      \end{align*}
      Observe that the only reduction rule that can be applied to the statement of the form \( b:= \RStr a \) is \rn{R-Read}.
      In order to apply this rule, we must check whether \( S_1(a) = n \cdot \seq m \), that is \( S_1(a) \) is not an empty string.
      Indeed, \( S_1(a) \) is non-empty because \( R_1, H, \Iseq \models \fml \) for some \( \Iseq \); recall that if \( \hd(\iset{a}) \) is undefined, then \( \fml \) is false.
      Therefore, we have \( \sconfig {R_1} {S_1} {b := \RStr a} \eval \sstate{R_1 \set{b \mapsto n}} {S_1 \set{a \mapsto \seq m}} \).
      It remains to show \( \relRHS{R}{H}{R_1 \set{b \mapsto n}}{S \set{a \mapsto \seq m}}{\ATE}{\fml'}{\Iseq'} \) for some \( \Iseq' \).
      For \( \Iseq' \), we take \( \Iseq \set{a \mapsto \tl(\Iseq(a))} \).
      We need to show
      \begin{align*}
        &R(x) =(R_1 \set{b \mapsto n})(x) \text{ for each $x\COL\INT\in \ATE$}\\
        &\Iset{H}{a'} \supseteq \setof{\Iseq'(a')} \mbox{ for each $a'$ such that $\ATE(a')\in\set{\RARR,\WARR}$}\\
        &\begin{aligned}
          &H(a',\Iseq'(a')[i]) = \Strm{\set{a \mapsto \seq m}}(a')[i] \\
          &\quad \mbox{ for each $a'$ such that $\ATE(a')\in\set{\RARR,\WARR}$, $0\le i < \length{\Iseq'(a')}$}\\
        &R_1,H, \Iseq' \models \fml'
          \end{aligned}.
      \end{align*}
      Most of these are straightforward consequence of \( \relRHS{R}{H}{R_1}{S_1}{\ATE}{\fml}{\Iseq} \).
      The first condition holds because, only the \( b \) part of \( R_1 \) has been updated and \( b \) does not have type \( \INT \).
      The second condition holds because \( \setof{\Iseq'(a')} \subseteq \setof{\Iseq(a')} \) for all \( a' : \RARR \in \ATE \).
      The third condition holds since
      \[
      H(a,\Iseq'(a)[i]) = H(a, \Iseq(a)[i + 1]) = S_1(a)[i + 1] = S_1\set{a \mapsto \seq m}(a)[i]
      \]
      for all \( 0 \le i < |\Iseq(a)| - 1 \), and \( \Iseq' \) is the same as \( \Iseq \) for addresses other than \( a \).
      Finally, the satisfaction relation holds thanks to Lemma~\ref{lem:assertion-subst-strm} and~\ref{lem:assertion-subst-index-strm} with the facts that \( \sem{a[\hd(\iset a)]}_{R_1, H, \Iseq} = n\) and \( \Iseq'  = \Iseq \set{a \mapsto \tl(\Iseq(a))} \).
    \item[\rn{Tr-InsWBuf}]
      In this case, we have \( \iexp = \WStr a b \), where \( \ATE(a) = \WARR \) and \( \ATE(b) = \BUF\), and
      \begin{align*}
        \fml = n \not\in \iset{a}\land a[n]=b\land [(\iset{a} \cdot n)/\iset{a}]\fml'.
      \end{align*}
      The only reduction rule that can be applied to the statement of the form \( \WStr a b \) is \rn{R-Write}.
      Hence, we have \( \sconfig {R_1} {S_1} {\WStr a b} \eval \sstate {R_1} {S_1\set{a \mapsto S_1(a) \cdot R_1(b)}} \) and this is the unique resulting state.
      We prove that \( \relRHS{R}{H}{R_1}{S_1\set{a \mapsto S_1(a) \cdot R_1(b)}}{\ATE}{\fml'}{\Iseq'}\) where \( \Iseq' \defeq \Iseq \set{a \mapsto \Iseq(a) \cdot n}\) and \( \Iseq \) is a chosen witness such that \(  \relRHS{R}{H}{R_1}{S_1}{\ATE}{\fml}{\Iseq} \).
      It should be noted that \( \Iseq' \) is well-defined in the sense that \( \Iseq'(a) \) is a sequence without any duplicated element because of the precondition \( n \notin \iset a \).
      The only nontrivial conditions to check to show that  \( \relRHS{R}{H}{R_1}{S_1\set{a \mapsto S_1(a) \cdot R_1(b)}}{\ATE}{\fml'}{\Iseq'}\)  are:
      \begin{gather*}
        \Iset{H}{a}  \supseteq \setof{\Iseq'(a)} \\
        H(a, \Iseq'(a)[\ell]) = S_1\set{a \mapsto S_1(a) \cdot R_1(b)}(a)[\ell] \text{where \( \ell = |\Iseq'| - 1\)} \\
        R_1,H, \Iseq \{ a \mapsto \Iseq(a) \cdot n \} \models \fml'.
      \end{gather*}
      (The other cases are trivial because the register file is untouched, and the only stream that was modified is that stored in the address \( a \).)
      To check the first condition, since \( \setof{\Iseq'(a)} = \setof{\Iseq(a)} \cup \{ n \} \) and \( \Iset{H}{a} \supseteq \setof{\Iseq(a)} \), it suffices to show that \( n \in \Iset{H}{a}  \).
      This holds because \( H, R_1, \Iseq \models b = a[n]\).
      In more detail, we have \( \sem{a[n]}_{H, R_1, \Iseq} = H(a, n) \) and \( \sem{b}_{R_1, H, \Iseq} = R_1(b) \), and thus \( H(a, n) = R_1(b) \).
      This equality allows us to prove the second condition because
      \[
       H(a, \Iseq'(a)[\ell]) = H(a, n) = R_1(b) = S_1\set{a \mapsto S_1(a) \cdot R_1(b)}(a)[\ell]
      \]
      The satisfaction relation holds by Lemma~\ref{lem:assertion-subst-index-strm}.

    \item[\rn{Tr-InsMove}]
      In this case, we have \( \iexp \equiv x := y \), where \( \ATE(x) = \BUF \) and \( \ATE(y) \in \set{\INT, \BUF} \), and \( \fml = \sub y x \fml'\).

      The rule \rn{R-Assign} is the only rule that can be applied to \( \sconfig {R_1} {S_1} {x := y} \), and thus, \( \sstate{R_1 \set{x \mapsto R_1(y)}} {S_1}\) is the only state such that \( \sconfig {R_1} {S_1} {x:=y} \eval \sstate{R_1 \set{x \mapsto R_1(y)}} {S_1} \).
      We need to check that \(  \relRH{R_1 \set{x \mapsto R_1(y)}}{S_1}{R_1}{S_1}{\ATE}{\fml'} \).
      Since \( \relRHS{R}{S}{R_1}{S_1}{\ATE}{\fml}{\Iseq} \), for some \( \Iseq \), and the stream pool is unchanged by the reduction, it suffices to check
      \begin{gather*}
        R(z) = R_1 \set{x \mapsto R_1(y)}(z) \text{ for each \( z : \INT \in \ATE \)}\\
        R_1 \set{x \mapsto R_1(y)}, H, \Iseq \models \fml'
      \end{gather*}
      The former holds because \( R(z) = R_1(z) \) for every \( z : \INT \in \ATE \) and \( x \) is not of type \( \INT \).
      The latter is a consequence of Lemma~\ref{lem:assertion-subst-strm} together with the facts \( R_1, H, \Iseq \models \sub y x \fml'\) and \( \sem{y}_{R_1, H, \Iseq} = R_1(y) \).
    \item[\rn{Tr-Conseq}]
      In this case, we have
      \begin{align*}
        \bjs{\ATE}{\fml_1}{\fml_1'}{\iexp}
      \end{align*}
      for some \( \fml_1 \) and \( \fml_1' \) satisfying
      \begin{align*}
        \models \fml \implies \fml_1 \quad \text{and} \quad \models \fml_1' \implies \fml'.
      \end{align*}
      Since \( \relRHS{R}{H}{R_1}{S_1}{\ATE}{\fml}{\Iseq} \) for some \( \Iseq \), we have \( R_1, H, \Iseq \models \fml \), and thus \( R_1, H, \Iseq \models \fml_1 \) from the above implication.
      Therefore, we have \( \relRHS{R}{H}{R_1}{S_1}{\ATE}{\fml_1}{\Iseq} \).
      By the induction hypothesis, there uniquely exists \( \sstate {R_2}{S_2} \) such that \( \sconfig {R_1}{S_1} \iexp \eval \sstate {R_2}{S_2} \).
      Moreover, we have \( \relRH R H {R_2} {S_2} \ATE {\fml_1'} \).
      It remains to show that \( \relRH R H {R_2} {S_2} \ATE {\fml'} \), but this is obvious from \( \relRH R H {R_2} {S_2} \ATE {\fml_1'} \) and \( \models \fml_1' \implies \fml' \).\qed
  \end{description}
\end{proof}

Now we prove the simulation relations for the stream translation.
\correctnessStream*
\begin{proof}[Proof of (1)]
  By induction on the derivation of \(\sconfig{R}{H}{s} \eval \sstate{R_1}{H_1} \) with a case analysis on the last rule used.

  In the proof, we also do a case analysis on the derivation of the translation.
  We assume, without loss of generality, that the derivation of the translation does not contain any use of \rn{Tr-InsertL}, \rn{Tr-InsertR} and \rn{Tr-Conseq}.
  For \rn{Tr-InsertL} and \rn{Tr-InsertR}, we can simply apply the Lemma~\ref{lem:buf-trans-insert-respects-invariants}.
  If we have  \( \bj \ATE \fml s \iexp {\fml'} \) and the last rule used is \rn{Tr-Conseq}, then we have
  \begin{gather*}
    \models \fml \implies \fml_1 \quad \models \fml_1' \implies \fml' \\
    \bj \ATE {\fml_1} s \iexp {\fml_1'}.
  \end{gather*}
  So, if \( \relRH R H {R'}{H'} \ATE \fml \), then from \( \models \fml \implies \fml_1 \), it is easy to see that we have \( \relRH R H {R'}{H'} \ATE {\fml_1} \).
  Similarly, if \( \relRH R H {R'}{H'} \ATE {\fml_1'} \) is shown, then we have \( \relRH R H {R'}{H'} \ATE {\fml_1} \).
  Hence, it suffices to prove our claim against \( \bj \ATE {\fml_1} s \iexp {\fml_1'} \), implying that we can ignore the use of \rn{Tr-Conseq}.

  Now we proceed to the case analysis.
  In what follows, we let \( \Iseq \) to be the witness of \( \relRH{R}{H}{R'}{S'}{\ATE}{\fml} \), i.e.~we assume that \( \Iseq \) has been chosen such that \( \relRHS{R}{H}{R'}{S'}{\ATE}{\fml}\Iseq \).
  We first prove the interesting cases, namely those that involve array accesses and the case for the kernel code.
  \begin{description}
    \item[\rn{R-ReadArray}]
      In this case, we have
      \[
      \infer
    {\econfig{R}{e}\eval m\andalso H(a,m) = n}
    {\sconfig{R}{H}{x:=a[e]} \eval \sstate{R\set{x\mapsto n}}{H}}
      \]
      By the inversion on the translation rule, we must also have
      \begin{gather*}
        \sj{\ATE}{b = a[e] \land [b/x]\fml'}{x := a[e]}{\fml'}{x := b} \\
        \ATE(a) = \RARR \quad \ATE(x)=\INT \quad  \ATE \vdash e:\INT \quad \ATE(b) = \BUF
      \end{gather*}
      where \( \fml \equiv b = a[e] \land [b/x]\fml' \) and \( t \equiv x:= b\).
      Observe that we have \( H(a, m) = n = R'(b) \) because \( R', H, \Iseq \models b = a[e] \).
      Thus, as the matching evaluation, we can take
      \[
      \infer
      {\econfig{R'}{b} \eval n}
      {\sconfig{R'}{\Strm'}{x:= b} \eval \sstate{R'\set{x\mapsto n}}{\Strm'}}
      \]
      using \rn{R-Read}.
      It remains to show that \( \relRHS{R \set{x \mapsto n}}{H}{R'\set{x\mapsto n}}{\Strm'}{\ATE}{\fml'} {\Iseq} \).
      Since \( H \) and \( \Strm' \) are not modified, the only condition we need to check is
      \[
      R'\set{x \mapsto n}, H, \Iseq \models \fml'.
      \]
      This follows from Lemma~\ref{lem:assertion-subst-strm} with \( R', H, \Iseq \models [b/x]\fml'
\) and \( \sem{a[e]}_{R', H, \Iseq} = n \).
    \item[\rn{R-WriteArray}]
      In this case, the last step of the derivation must be of the form:
      \begin{align*}
        \infer{\econfig{R}{e}\eval m} 
    {\sconfig{R}{H}{a[e]:=x} \eval \sstate{R}{H\set{\Addr{a}{m}\mapsto R(x)}}}
      \end{align*}
      By the inversion on the translation rule, we also have
      \begin{gather*}
        \bj{\ATE}{e\not\in \iset{a}\land [x/b, \updateAr{a}{e}{x}/a]\fml'}{a[e] := x}{\fml'}{b := x}\\
        \fml = e\not\in \iset{a}\land [x/b, \updateAr{a}{e}{x}/a]\fml' \\
        \ATE(x)=\INT \quad \ATE(a)=\WARR \quad \ATE\p e:\INT \quad \ATE(b)=\BUF.
      \end{gather*}
      As the matching transition, we take \( \sconfig {R'} {H'} {b := x} \eval \sstate {R' \set{b \mapsto R(x)}} {H'}\).
      We are left to show that \( \relRHS R {H \set{(a, m) \mapsto R(x)}} {R' \set{b \mapsto R(x)}} {S'} \ATE {\fml'} \Iseq \).
      In particular, we need to show that
      \begin{align*}
      &H(a,\Iseq(a)[i]) = \Strm'(a)[i] \mbox{ for each $a$ such that $\ATE(a)\in\set{\RARR,\WARR}$, $0\le i < \length{\Iseq(a)}$}\\
       & R' ,H, \Iseq \models \fml'.
      \end{align*}
      The other conditions are straightforward consequence of \( \relRHS R H {R'} {S'} \ATE \fml \Iseq \) because \( R' \) is modified only for a buffer variable and the domain of the heap has gotten bigger while \( \Iseq \) remains the same.
      To show that the condition on the relation between \( H \) and \( S' \) holds, it suffices to show that \( m \notin \Iseq(a) \).
      We have \( m \notin \Iseq(a) \) because \( R', H, \Iseq \models e \notin \iset a \), \( \sem{e}_{R', H, \Iseq} = m\) and \( \sem{\iset a}_{R, H, \Iseq} = \Iseq(a)\).
      The satisfaction relation \( R' ,H, \Iseq \models \fml' \) holds by applying Lemma~\ref{lem:assertion-subst-strm} and~\ref{lem:assertion-subst-arr-strm} to \( R', H, \Iseq \models [x/b, \updateAr{a}{e}{x}/a]\fml' \).

    \item[\rn{R-CallKer}]
      For this case, we just need to check that \( \relRH R H {R'} {\Strm'} \ATE \fml \) implies  \( \relRH R H {R'} {\Strm'} {\fliparray \ATE} \fml \);the rest of the proof is a straightforward application of the induction hypothesis.
      This indeed holds because the two conditions
      \begin{align*}
        &\Iset{H}{a} \supseteq \setof{\Iseq(a)} \mbox{ for each $a$ such that $\ATE(a)\in\set{\RARR,\WARR}$}\\
        &H(a,\Iseq(a)[i]) = \Strm'(a)[i] \mbox{ for each $a$ such that $\ATE(a)\in\set{\RARR,\WARR}$, $0\le i < \length{\Iseq(a)}$}
      \end{align*}
      of \( \relRH R H {R'} {\Strm'} \ATE \fml \) involving the array types are symmetric between \( \RARR \) and \( \WARR \).
  \end{description}

  Now we prove the remaining cases, which are rather standard.
  \begin{description}
      \item[\rn{R-Assign}]
      Similar to the case for \rn{R-ReadArray}.
      \sk{I might write the details if I have time.}
    \item[\rn{R-Seq}]
      The last step of the derivation must be of the form:
      \[
      \infer{\sconfig{R}{H}{s_1}\eval \sstate{R_2}{H_2}\andalso
         \sconfig{R_2}{H_2}{s_2}\eval \sstate{R_1}{H_1}}
    {\sconfig{R}{H}{s_1;s_2} \eval \sstate{R_1}{H_1}}
      \]
      where \( s \equiv s_1; s_2\).
      By inversion on the translation rule, we must also have
      \begin{align*}
        \bj{\ATE}{\fml}{s_1}{\fml''}{\iexp_1}
        \quad \text{and} \quad
        \bj{\ATE}{\fml''}{s_2}{\fml'}{\iexp_2}
        \quad \text{with } \iexp \equiv \iexp_1; \iexp_2
      \end{align*}
      Since \( \relRH R H {R'} {S'} \ATE \fml \) from the assumption, we have
      \begin{align*}
        \sconfig{R'}{S'}{\iexp_1}\eval \sstate{R'_2}{S'_2}
        \quad \text{and} \quad
        \relRH{R_2}{H_2}{R'_2}{S'_2}{\ATE}{\BTE''}
      \end{align*}
      for some \( R'_2 \) and \( S'_2 \) by the induction hypothesis.
      By applying the induction hypothesis, this time to \( \sconfig{R_2}{H_2}{s_2}\eval \sstate{R_1}{H_1} \), we get
      \begin{align*}
        \sconfig{R'_2}{S'_2}{\iexp_2}\eval \sstate{R'_1}{S'_1}
        \quad \text{and} \quad
        \relRH{R_1}{H_1}{R'_1}{S'_1}{\ATE}{\BTE'}
      \end{align*}
      for some \( R'_1 \) and \( S'_1 \).
    \item[\rn{R-IfTrue}]
      In this case, we must have
      \[
      \infer{R(x) = 0\andalso \sconfig{R}{H}{s_1}\eval \sstate{R_1}{H_1}}
      {\sconfig{R}{H}{\ifexp{x}{s_1}{s_2}} \eval \sstate{R_1}{H_1}}
      \]
      where \( s \equiv \ifexp{x}{s_1}{s_2} \).
      Thus, we must also have
      \begin{align*}
        \bj{\ATE}{{\color{magenta} \fml\land x=0}}{s_1}{\fml'}{\iexp_1}
        \quad \text{and} \quad
        \bj{\ATE}{{\color{magenta} \fml \land x\ne 0}}{s_2}{\fml'}{\iexp_2}
      \end{align*}
      for \( \iexp \equiv \ifexp x {\iexp_1} {\iexp_2} \).
      By Lemma~\ref{lem:relRH-preserves-eval-expr}, we have \( \econfig {R'} x \eval 0 \), and thus, \( R', H, \Iseq \models x = 0 \).
      This together with \( \relRHS {R} H {R'} {S'} \ATE \fml \Iseq \) implies that \( \relRHS R H {R'} {S'} \ATE {\fml \land x = 0} \Iseq\).
      Hence, by the induction hypothesis, there exists \( \sstate {R'_1} {S'_1} \) such that \( \sconfig {R'} {H'} {\iexp_1} \eval \sstate {R'_1} {S'_1} \) and \( \relRH {R_1} {H_1} {R'_1} {S'_1} \ATE {\fml'}\).
      By \rn{R-IfTrue}, we also have \( \sconfig{R_1}{H_1} \iexp \eval \sstate {R'_1} {S'_1}\) as desired.
    \item[\rn{R-IfFalse}]
      Similar to the previous case.
    \item[\rn{R-ForExit}]
      In this case, we have
      \begin{align*}
        \infer{\econfig{R}{e}\eval k \andalso k\ge m}
        {\sconfig{R}{H}{\forexp{x}{e}{m}{n}{s_1}}\eval \sstate{R \set{x \mapsto k}}{H}}
      \end{align*}
      with \( s \equiv \forexp{x}{e}{m}{n}{s_1} \).
      By the inversion on the translation rule, we have
      \begin{align*}
        {\bj{\ATE}{[e/x]\fml''}
        {s}{\fml''  \land x\ge m}{\forexp{x}{e}{m}{n}{\iexp_1}}}
      \end{align*}
      where \( \iexp \equiv \forexp{x}{e}{m}{n}{\iexp_1} \) and \( \fml = [e/x]\fml'' \).
      Since \( \econfig{R}{e}\eval k \), we also have \( \econfig{R'}{e}\eval k \) by Lemma~\ref{lem:relRH-preserves-eval-expr}.
      Therefore, we have \( \sconfig {R'} {S'} \iexp \eval \sstate {R' \set{x \mapsto k}}{S'} \) using \rn{R-ForExit}.
      We can check that \( \relRHS{R \set{x \mapsto k}}{H}{R' \set{x \mapsto k}}{S'} \ATE {\fml'' \land x \ge m} \Iseq \) as in the case of \rn{R-ReadArray} and \rn{R-Assign}.
      The only difference is that, in this case, we additionally need to check \( R' \set{x \mapsto k}, H, \Iseq \models x \ge m \), but this trivially holds because \( R' \set{x \mapsto k}(x) = k \ge m \).
    \item[\rn{R-ForLoop}]
      In this case, the last rule of the derivation is
      \begin{equation*}
      \frac{\splitfrac{\econfig{R}{e}\eval k \andalso k < m \quad
          \sconfig{R\set{x\mapsto k}}{H}{s'}\eval \sstate{R_1}{H_1}}{
        s' \equiv {s_1; \forexp{x}{x+n}{m}{n}{s_1}}}}
    {\sconfig{R}{H}{\forexp{x}{e}{m}{n}{s_1}}\eval \sstate{R_1}{H_1}}
      \end{equation*}
      with \( s \equiv \forexp{x}{e}{m}{n}{s_1} \).
      Observe that this means that there exists a state \( \sstate {R_2}{H_2} \) such that
      \( \sconfig {R\set{x \mapsto k}} {H} {s_1} \eval \sstate {R_2} {H_2} \) and \( \sconfig {R_2} {H_2} {\forexp{x}{x+n}{m}{n}{s_1}} \eval \sstate {R_1} {H_1}\).
      By the inversion on the translation rule, we have
      \begin{align}
        \bj{\ATE}{{\color{magenta}\fml'' \land x < m}}{s_1}{[x+n/x]\fml''}{\iexp_1} \label{eq:lem:correctness-buf:forloop}
      \end{align}
      for some \( \iexp_1 \) and \( \fml'' \) such that
      \begin{gather*}
      \iexp \equiv \forexp{x}{e}{m}{n}{\iexp_1},\\
      \fml = [e/x]\fml'', \text{ and }  \fml' = {\color{magenta}\fml'' \land x \ge m}.
      \end{gather*}
      It is easy to check that we have \( \relRH{R \set{x \mapsto k}}{H}{R'\set{x \mapsto k}}{S'} \ATE {\fml'' \land x < m} \).
      \sk{TODO: write the details}
      By applying the induction hypothesis to  \( \sconfig {R\set{x \mapsto k}} {H} {s_1} \eval \sstate {R_2} {H_2} \), we obtain a state \( \sstate {R'_2} {S'_2} \) such that \( \sconfig {R' \set{x \mapsto k}} {S'} {\iexp_1}\eval \sstate {R'_2} {S'_2} \) and \( \relRH {R_2} {H_2} {R'_1} {S'_2} \ATE { \sub {x + n} x \fml''} \).
      Now we apply the induction hypothesis to \( \sconfig {R_2} {H_2} {\forexp{x}{x+n}{m}{n}{s_1}} \eval \sstate {R_1} {H_1}\).
      To this end, we first show that
      \begin{align*}
        \ATE \vdash \{\sub {x + n} x \fml''\} &{\forexp{x}{x+n}{m}{n}{s_1}}\\
                                                &\Longrightarrow{\forexp{x}{x+n}{m}{n}{u_1}} \{\fml'\}
      \end{align*}
      This is derivable by applying \rn{R-For} to \eqref{eq:lem:correctness-buf:forloop}.
      We can now apply the induction hypothesis, and obtain a state \( \sstate {R_1'} {S_1'} \) such that \( \sconfig {R'_2} {H'_2} {\forexp{x}{x + n}{m}{n}{\iexp_1}} \eval \sstate {R'_1} {S'_1} \) and \( \relRH {R_1} {H_1} {R'_1} {S'_1} \ATE {\fml'' \land x \ge m }\).
      Using the rule \rn{R-ForLoop}, we also have \( \sconfig {R'} {S'} {\iexp} \eval \sstate {R'_1} {S'_1} \) as desired.

  \end{description}
\qed
\end{proof}
We only sketch the proof for the opposite direction since the proof is essentially the same.
\begin{proof}[Proof Sketch of (2)]
  By induction on the derivation of \(\sconfig{R'}{S'}{u} \eval \sstate{R'_1}{S'_1} \) with a case analysis on the last rule used.
  The nontrivial cases are the case when the rules \rn{R-Read} or \rn{R-Write} are used or a buffer is accessed.
  Since the proof for the other cases are identical to those of (1) we omit these cases.
  That is, if the last rule used in the derivation are not the above mentioned rules, then we can use the same rule for the source program to obtain a matching transition; proving the relation \( \relRH {R_1}{H_1}{R'_1}{\Strm'_1} \ATE \fml \) is done as in the proof of (1).

  If \rn{R-Read} or \rn{R-Write} are the last rule used, then \( u \) is of the form \(x := \RStr a \) or \(  \WStr a x \).
  In these cases, the last rule applied must be (if we ignore the consequence rule) \rn{Tr-InsertL} or \rn{Tr-InsertR}.
  These cases follow from Lemma~\ref{lem:buf-trans-insert-respects-invariants}.

  The remaining interesting case is the case \rn{R-Assign}.
  That is, the case where the last rule applied is of the form:
  \begin{align*}
    \infer{\econfig{R'}{e}\eval n}
    {\sconfig{R'}{S'}{x:=e} \eval \sstate{R'\set{x\mapsto n}}{S'}}
  \end{align*}
  By the definition of the translation rules, the only cases where this rule involves a buffer are the when (i) \( x : \BUF \) or (ii) \( e \) is a buffer \( b \) and \( x : \INT \).

  We first consider the case (i).
  In this case let us rename \( x \) to \( b \) so that it is clear that the variable is of type \( \BUF \).
  By inversion on the transformation rule, we must have
  \begin{gather*}
    e \equiv y \quad  \ATE(y)=\INT \\
    u \equiv b:= y \quad s \equiv a[e'] := y  \quad \ATE(a)=\WARR \quad \ATE\p e':\INT \\
    {\bj{\ATE}{e'\not\in \iset{a}\land [y/b, \updateAr{a}{e'}{y}/a]\fml'}{a[e'] := y}{\fml'}{b := y}} \\
    \fml \equiv e'\not\in \iset{a}\land [y/b, \updateAr{a}{e'}{y}/a]\fml'.
  \end{gather*}
  for some \( y \) and \( e' \).
  As the matching reduction we can take \( \sconfig R H {s} \eval \sstate{R}{H \set{(a,m) \mapsto R(y)}}\), where \( \econfig R {e'} \eval m \), using \rn{R-WriteArray}.
  We are left to check that \( \relRH {R} {H \set{(a,m) \mapsto R(y)}} {R'\set{b \mapsto R'(y)}} {S'} \ATE {\fml'}\) under the assumption that \( \relRH R H {R'} {S'} \ATE \fml \).
  We have already showed this in the proof of (1) (the case for \rn{R-WriteArray}).

  Now we consider the case (ii).
  In this case, the last transformation rule applied must be \rn{Tr-READMem}, and thus, we must have
    \begin{gather*}
      u \equiv x := b \quad s \equiv x := a[e']  \\
      \ATE(x)= \INT \quad \ATE(a)=\RARR \quad \ATE\p e':\INT \quad \ATE(b)=\BUF \\
      \bj{\ATE}{b=a[e]\land [b/x]\fml'}{x := a[e']}{\fml'}{x := b} \\
      \fml \equiv b=a[e]\land [b/x]\fml'
  \end{gather*}
    We can take \( \sconfig R H s \eval \sstate {R \set{R \mapsto H(a, m)}} H \) as the matching transition.
    Here, we used the rule \rn{R-ReadArray} and \( m \) is the evaluation result of \( e' \).
    It remains to show that \( \relRH {R \set{ x \mapsto H(a, m) }} H {R'\set{x \mapsto R'(b)}} {S'} \ATE {\fml'}\) under the assumption that \( \relRH R H {R'} {S'} \ATE \fml \).
    Again, this is something we already proved in the proof of (1) (the case for \rn{R-ReadArray}).
    \qed
  \end{proof}

\fi
\end{document}